\documentclass[10pt,twoside,twocolumn,english,aps,manuscript,nofootinbib, nobibnotes, secnumarabic]{revtex4}
\usepackage[T1]{fontenc}
\usepackage[latin9]{inputenc}
\usepackage[letterpaper]{geometry}
\usepackage{xcolor}
\usepackage{units}
\usepackage{amsthm}
\usepackage{amsmath}
\usepackage{amssymb}
\usepackage[all]{xy}
\PassOptionsToPackage{normalem}{ulem}
\usepackage{ulem}

\makeatletter

\providecommand{\tabularnewline}{\\}
\providecolor{lyxadded}{rgb}{0,0,1}
\providecolor{lyxdeleted}{rgb}{1,0,0}

\@ifundefined{textcolor}{}
{%
 \definecolor{BLACK}{gray}{0}
 \definecolor{WHITE}{gray}{1}
 \definecolor{RED}{rgb}{1,0,0}
 \definecolor{GREEN}{rgb}{0,1,0}
 \definecolor{BLUE}{rgb}{0,0,1}
 \definecolor{CYAN}{cmyk}{1,0,0,0}
 \definecolor{MAGENTA}{cmyk}{0,1,0,0}
 \definecolor{YELLOW}{cmyk}{0,0,1,0}
}
\theoremstyle{plain}
\newtheorem{thm}{\protect\theoremname}[section]
  \theoremstyle{definition}
  \newtheorem{defn}[thm]{\protect\definitionname}
  \theoremstyle{plain}
  \newtheorem{criterion}[thm]{\protect\criterionname}
  \theoremstyle{plain}
  \newtheorem{cor}[thm]{\protect\corollaryname}
  \theoremstyle{definition}
  \newtheorem{example}[thm]{\protect\examplename}

\usepackage[all]{xy}

\usepackage{multirow}
\usepackage{pslatex}
\usepackage{calligra}
\DeclareMathAlphabet{\mathcalligra}{T1}{calligra}{m}{n}
\DeclareFontShape{T1}{calligra}{m}{n}{<->s*[2.2]callig15}{}

\def\frontmatter@abstractheading{}

\makeatother

\usepackage{babel}
  \providecommand{\corollaryname}{Corollary}
  \providecommand{\criterionname}{Criterion}
  \providecommand{\definitionname}{Definition}
  \providecommand{\examplename}{Example}
\providecommand{\theoremname}{Theorem}

\begin{document}

\title{Selectivity in Probabilistic Causality:\\Where Psychology Runs Into
Quantum Physics}

\author{Ehtibar N. Dzhafarov}

\thanks{Corresponding author: Ehtibar Dzhafarov, Purdue University, Department
of Psychological Sciences, 703 Third Street West Lafayette, IN 47907,
USA. email: ehtibar@purdue.edu.}

\affiliation{Purdue University}

\author{Janne V. Kujala}

\affiliation{University of Jyväskylä}
\begin{abstract}
\mbox{}

Given a set of several inputs into a system (e.g., independent variables
characterizing stimuli) and a set of several stochastically non-independent
outputs (e.g., random variables describing different aspects of responses),
how can one determine, for each of the outputs, which of the inputs
it is influenced by? The problem has applications ranging from modeling
pairwise comparisons to reconstructing mental processing architectures
to conjoint testing. A necessary and sufficient condition for a given
pattern of selective influences is provided by the Joint Distribution
Criterion, according to which the problem of ``what influences what''
is equivalent to that of the existence of a joint distribution for
a certain set of random variables. For inputs and outputs with finite
sets of values this criterion translates into a test of consistency
of a certain system of linear equations and inequalities (Linear Feasibility
Test) which can be performed by means of linear programming. While
new in the behavioral context, both this test and the Joint Distribution
Criterion on which it is based have been previously proposed in quantum
physics, in dealing with generalizations of Bell inequalities for
the quantum entanglement problem. The parallels between this problem
and that of selective influences in behavioral sciences are established
by observing that noncommuting measurements in quantum physics are
mutually exclusive and can therefore be treated as different levels
of one and the same factor.

\mbox{}

\textsc{Keywords:} Bell-type inequalities, EPR paradigm, factorial
design, Fine's inequalities, joint distribution criterion, probabilistic
causality, mental architectures, random outputs, selective influences,
quantum entanglement, Thurstonian scaling.

\markboth{Dzhafarov and Kujala}{Selectivity in Probabilisitc Causality}
\end{abstract}
\maketitle

\section{Introduction\label{sec:Introduction}}

This paper deals with \emph{diagrams of selective influences}, like
this one: \emph{
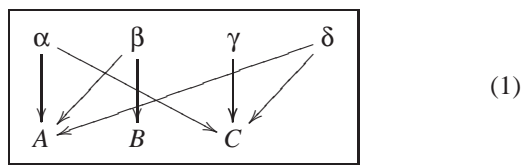
\begin{equation}
\boxed{\begin{array}{c}
\xymatrix{\alpha\ar[d]\ar[drr] & \beta\ar[d]\ar[dl] & \gamma\ar[d]\ar[d] & \delta\ar[dl]\ar[dlll]\\
A & B & C
}
\end{array}}\label{diag:1}
\end{equation}
}

\protect{\noindent}The Greek letters in this diagram represent \emph{inputs},
or \emph{external factors}, e.g., parameters of stimuli whose values
can be chosen at will or observed and recorded. The capital Roman
letters stand for random outputs characterizing reactions of the system
(an observer, a group of observers, stock market, a set of photons,
etc.). The arrows show which factor influences which random output.
The factors are treated as \emph{deterministic} entities: even if
$\alpha,\beta,\gamma,\delta$ in reality vary randomly (e.g., being
randomly generated by a computer program, or being concomitant parameters
of observations, such as age of respondents), for the purposes of
analyzing selective influences the random outputs $A,B,C$ are always
viewed as \emph{conditioned} upon various combinations of specific
values of $\alpha,\beta,\gamma,\delta$. The first question to ask
is: what is the meaning of the above diagram if the random outputs
$A,B,C$ in it are not necessarily stochastically independent? (If
they are, the answer is of course trivial.) And once the meaning of
the diagram of selective influences is established, how can one determine
that this diagram correctly characterizes the dependence of the joint
distributions of the random outputs $A,B,C$ on the external factors
$\alpha,\beta,\gamma,\delta$? 

These questions are important, because the assumption of stochastic
independence of the outputs more often than not is either demonstrably
false or adopted for expediency alone, with no other justification.
At the same time the assumption of selectivity in causal relations
between inputs and stochastic outputs is ubiquitous in theoretical
modeling, often being built in the very language of the models. For
instance, in Thurstone's most general model of pairwise comparisons
(Thurstone, 1927) it is assumed that each of the two stimuli is mapped
into ``its'' internal representation, while the two representations
are stochastically interdependent random entities. In Dzhafarov (2003),
Dzhafarov and Gluhovsky (2006), and Kujala and Dzhafarov (2008) the
reader may find other motivating applications for the notion of selective
influences: same-different comparisons, conjoint testing, parallel-serial
networks of mental operations, response time decompositions, and all
conceivable combinations of regression analysis and factor analysis.
In this paper we add another motivating example, the quantum entanglement
problem in quantum physics. 

This paper continues and expands the analysis of selective influences
presented in Dzhafarov and Kujala (2010). The familiarity with it
can be helpful, but the main concepts, terminology, and notation are
recapitulated in Section \ref{sec:Basic-Notions}. Unlike in Dzhafarov
and Kujala (2010), however, here we do not pursue the goal of maximal
generality of formulations, focusing instead on the conceptual set-up
that applies to commonly encountered experimental designs. This means
a finite number of factors, each having a finite number of values.
It also means that the random outcomes influenced by these factors
are \emph{random variables} in the narrow sense of the word: their
values are vectors of real numbers or elements of countable sets,
rather than more complex structures, such as functions or sets. This
is done primarily to simplify and shorten exposition, and also because
the \emph{Linear Feasibility Test}, a new (for behavioral sciences)
application of the Joint Distribution Criterion on which we focus
in this paper (Section \ref{sec:Linear-Feasibility-Test}), is confined
to finite sets of finite-valued factors and finite-valued random variables.
This also allows us to emphasize a simple but important and previously
overlooked proposition, Theorem \ref{thm:In-Definition-,}, which
essentially says that, when dealing with observable random variables,
the unobservable random entities of the theory can also be assumed
to be random variables (in the narrow sense). In another respect,
however, the present treatment is more general than that in Dzhafarov
and Kujala (2010): we allow for \emph{incomplete designs}, those in
which some but not necessarily all combinations of the values of the
factors serve as allowable treatments. This modification is critical
for the possibility of representing any diagram of selective influences,
such as (\ref{diag:1}), in a \emph{canonical} \emph{form}, with every
random output being selectively influenced by one and only one factor. 

As it turns out, both the Linear Feasibility Test and the Joint Distribution
Criterion on which it is based have their analogues in quantum physics.%
\footnote{\label{fn:We-are-grateful}We are grateful to Jerome Busemeyer of
Indiana University who pointed out to us that the formulation of the
Joint Distribution Criterion in our earlier work has the same formal
structure as the identically titled criterion in Fine (1981a-b), in
his analysis of quantum entanglement.%
} To appreciate the analogy, however, one has to adopt the interpretation
of noncommuting quantum measurements performed on a given component
of a quantum-entangled system as mutually exclusive factor levels
of the same factor. In Sections \ref{sec:Quantum-Entanglement-and}
and \ref{sec:Linear-Feasibility-Test} we discuss the parallels between
the existence of a classical explanation for an entanglement situation
in quantum mechanics and the adherence of a behavioral experiment
to a diagram of selective influences. 

The term ``test'' in this paper is used in the meaning of necessary
(sometimes necessary and sufficient) conditions for diagrams of selective
influences. The usage is the same as when we speak of the tests for
convergence in calculus or for divisibility in arithmetic. That is,
the meaning of the term is non-statistical. We assume that random
outputs are known on the population level. General considerations
related to statistical tests based on our population level tests are
discussed in Section \ref{sub:Sample-level-tests}, but specific statistical
issues are outside the scope of this paper.

\section{Basic Notions\label{sec:Basic-Notions}}

In this section, we establish the terminology, notation, and recapitulate
basic facts related to factors, random variables, and the dependence
of the latter on the former. We follow Dzhafarov and Kujala (2010),
adding observations related to the factorial designs being incomplete
and random outputs being random variables in the narrow sense of the
term. At the end of the section we discuss the parallels between the
issue of selective influence in behavioral sciences and the quantum
entanglement problem.

\subsection{Factors, factor points, treatments}

A \emph{factor} $\alpha$ is treated as a set of \emph{factor points},
each of which has the format ``value (or level) $x$ of factor $\alpha$.''
In symbols, this can be presented as $\left(x,`\alpha\textnormal{'}\right)$,
where $`\alpha\textnormal{'}$ is the unique name of the set $\alpha$
rather than the set itself. It is convenient to write $x^{\alpha}$
in place of $\left(x,`\alpha\textnormal{'}\right)$. Thus, if a factor
with the name $`intensity\textnormal{'}$ has three levels, $`low,\textnormal{'}$
$`medium,\textnormal{'}$ and $`high,\textnormal{'}$ then this factor
is taken to be the set 
\[
intensity=\left\{ low^{intensity},medium^{intensity},high^{intensity}\right\} .
\]
There is no circularity here, for, say, the factor point $low^{intensity}$
stands for $\left(value=low,name=`intensity\textnormal{'}\right)$
rather than $\left(value=low,set=intensity\right)$.

We will deal with finite sets of factors $\Phi=\left\{ \alpha_{1},\ldots,\alpha_{m}\right\} $,
with each factor $\alpha\in\Phi$ consisting of a finite number of
factor points, 
\[
\alpha=\left\{ v_{1}^{\alpha},\ldots,v_{k_{\alpha}}^{\alpha}\right\} .
\]
Clearly, $\alpha\cap\beta=\varnothing$ for any distinct $\alpha,\beta\in\Phi$. 

A \emph{treatment}, as usual, is defined as the set of factor points
containing one factor point from each factor, 
\[
\phi=\left\{ x_{1}^{\alpha_{1}},\ldots,x_{m}^{\alpha_{m}}\right\} \in\alpha_{1}\times\ldots\times\alpha_{m}.
\]

The \emph{set of treatments} (used in an experiment or considered
in a theory) is denoted by $T\subset\alpha_{1}\times\ldots\times\alpha_{m}$
and assumed to be nonempty. Note that $T$ need not include all possible
combinations of factor points. This is an important consideration
in view of the ``canonical rearrangement'' described below. Also,
incompletely crossed designs occur broadly --- in an experiment because
the entire set $\alpha_{1}\times\ldots\times\alpha_{m}$ may be too
large, or in a theory because certain combinations of factor points
may be physically or logically impossible (e.g., contrast and shape
cannot be completely crossed if zero is one of the values for contrast).

\subsection{Random variables}

We assume the reader is familiar with the notion of a random entity
(random variable in the general sense of the term) $A$ associated
with an observation space $\left(\mathcal{A},\Sigma\right)$, where
$\mathcal{A}$ is the set of possible values for $A$, and $\Sigma$
a sigma-algebra (set of events) on $\mathcal{A}$. A \emph{random
variable} (in the narrow sense) is a special case of a random entity,
defined as follows: 

(i) if $\mathcal{A}$ is countable, $\Sigma$ is the power set of
$\mathcal{A}$, then $A$ is a random variable; 

(ii) if $\mathcal{A}$ is an interval of reals, $\Sigma$ is the Lebesgue
sigma-algebra on $\mathcal{A}$, then $A$ is a random variable; 

(iii) if $A_{1},\ldots,A_{n}$ are random variables, then any jointly
distributed vector $\left(A_{1},\ldots,A_{n}\right)$ whose observation
space is the conventionally understood product of the observations
spaces for $A_{1},\ldots,A_{n}$ is a random variable. 

We use the relational symbol $\sim$ in the meaning of ``is distributed
as.'' $A\sim B$ is well defined irrespective of whether $A$ and
$B$ are jointly distributed. 

Let, for each treatment $\phi\in T$, there be a vector of jointly
distributed random variables $A=(A_{1},\ldots,A_{n})$ with a fixed
(product) observation space and the probability measure $\mu_{\phi}$
that depends on $\phi$.%
\footnote{\label{fn:Invariance}The convenient assumption of the invariance
of the observation space for $A$ with respect to $\phi$ is innocuous:
one can always redefine the observation spaces for different treatments
$\phi$ to make them coincide. %
} Then we say that we have a \emph{vector of jointly distributed random
variables that depends on treatment} $\phi$, and write 
\[
A(\phi)=(A_{1},\ldots,A_{n})(\phi),\quad\phi\in T.
\]
A correct way of thinking of $A(\phi)$ is that it represents a \emph{set
of vectors of jointly distributed random }variables, each of these
vectors being labeled (indexed) by a particular treatment. Any subvector
of $A\left(\phi\right)$ should also be written with the argument
$\phi$, say, $(A_{1},A_{2},A_{3})\left(\phi\right)$. If $\phi$
is explicated as $\phi=\left\{ x_{1}^{\alpha_{1}},\ldots,x_{m}^{\alpha_{m}}\right\} $,
we write $A(\phi)=A(x_{1}^{\alpha_{1}},\ldots,x_{m}^{\alpha_{m}})$.

It is important to note that for distinct treatments $\phi_{1}$ and
$\phi_{2}$ the corresponding $A(\phi_{1})$ and $A(\phi_{2})$ \emph{do
not possess a joint distribution}, they are \emph{stochastically unrelated}.
This is easy to understand: since $\phi_{1}$ and $\phi_{2}$ are
mutually exclusive conditions for observing values of $A$, there
is no non-arbitrary way of choosing which value $a=(a_{1},\ldots,a_{n})$
observed at $\phi_{1}$ should be paired with which value $a'=(a'_{1},\ldots,a'_{n})$
observed at $\phi_{2}$. To consider $A(\phi_{1})$ and $A(\phi_{2})$
stochastically independent and to pair every possible value of $A(\phi_{1})$
with every possible value $A(\phi_{2})$ is as arbitrary as, say,
to consider them positively correlated and to pair every quantile
of $A(\phi_{1})$ with the corresponding quantile of $A(\phi_{2})$.

\subsection{Arrow diagrams, canonically (re)arranged}

Given a set of factors $\Phi=\left\{ \alpha_{1},\ldots,\alpha_{m}\right\} $
and a vector $A(\phi)=(A_{1},\ldots,A_{n})(\phi)$ of random variables
depending on treatment, an \emph{arrow diagram} is a mapping
\begin{equation}
M:\left\{ 1,\ldots,n\right\} \rightarrow2^{\Phi}\label{eq:DSI}
\end{equation}
($2^{\Phi}$ being the set of subsets of $\Phi$). Later, in Definition
\ref{def:(SI,-bijective)}, the arrows will be interpreted as indicating
selective influences, but for now this is unimportant. The set 
\[
\Phi_{i}=M\left(i\right),\;(i=1,\ldots,n),
\]
is referred to as the subset of factors \emph{corresponding to} $A_{i}$.
It determines, for any treatment $\phi\in T$, the subtreatments $\phi_{\Phi_{i}}$
defined as
\[
\phi_{\Phi_{i}}=\left\{ x^{\alpha}\in\phi:\alpha\in\Phi_{i}\right\} ,\quad i=1,\ldots,n.
\]
Subtreatments $\phi_{\Phi_{i}}$ across all $\phi\in T$ can be viewed
as \emph{admissible} \emph{values} of the subset of factors $\Phi_{i}$
($i=1,\ldots,n$). Note that $\phi_{\Phi_{i}}$ is empty whenever
$\Phi_{i}$ is empty. 

The simplest arrow diagram is \emph{bijective}, with correspondences
\begin{equation}
\boxed{\xymatrix{\alpha_{1}\ar@{->}[d] & \ldots & \alpha_{n}\ar@{->}[d]\\
A_{1} & \ldots & A_{n}
}
}.\label{diag:bijective DSI}
\end{equation}

\protect{\noindent}We can simplify the subsequent discussion without
sacrificing generality by agreeing to reduce each arrow diagram (in
the context of selective influences) to a bijective form, by appropriately
redefining factors and treatments. It is obvious how this should be
done. Given the subsets of factors $\Phi_{1}\ldots,\Phi_{n}$ determined
by an arrow diagram (\ref{eq:DSI}), each $\Phi_{i}$ can be viewed
as a factor identified with the set of factor points 
\[
\alpha_{i}^{*}=\left\{ (\phi_{\Phi_{i}})^{\alpha_{i}^{*}}:\phi\in T\right\} ,
\]
in accordance with the notation we have adopted for factor points:
$(\phi_{\Phi_{i}})^{\alpha_{i}^{*}}=(\phi_{\Phi_{i}},`\alpha^{*}\textnormal{'})$.
If $\Phi_{i}$ is empty, then $\phi_{\Phi_{i}}$ is empty too, and
the factor $\alpha_{i}^{*}$ consists of only the dummy factor point
$\varnothing^{\alpha_{i}}$ (where $\varnothing$ denotes the empty
set). The set of treatments $T$ for the original factors $\{\alpha_{1},\ldots,\alpha_{m}\}$
should then be redefined for the vector of new factors $(\alpha_{1}^{*},\ldots,\alpha_{n}^{*})$
as
\[
T^{*}=\left\{ \left\{ (\phi_{\Phi_{1}})^{\alpha_{1}^{*}},\ldots,(\phi_{\Phi_{n}})^{\alpha_{n}^{*}}\right\} :\phi\in T\right\} \subset\alpha_{1}^{*}\times\ldots\times\alpha_{n}^{*}.
\]
We call this (re)definition of factor points, factors, and treatments
the \emph{canonical (re)arrangement.} We can say that the random variables
following canonical (re)arrangement can be indexed by the corresponding
factors. Thus, when convenient, we can write in (\ref{diag:bijective DSI})
$A_{\left\{ \alpha_{1}\right\} }$ in place of $A_{1}$, $A_{\left\{ \alpha_{2}\right\} }$
in place of $A_{2},$ etc. The notation $\phi_{\Phi_{i}}=\phi_{\left\{ \alpha_{i}\right\} }$
then indicates the singleton set $\left\{ x^{\alpha_{i}}\right\} \subset\phi$.
As usual, we write $x^{\alpha_{i}}$ in place of $\left\{ x^{\alpha_{i}}\right\} $:
\[
\phi_{\left\{ \alpha_{i}\right\} }=\left\{ x_{1}^{\alpha_{1}},\ldots,x_{n}^{\alpha_{n}}\right\} _{\left\{ \alpha_{i}\right\} }=x_{i}^{\alpha_{i}}.
\]

\subsection{The criterion}
\begin{defn}[\emph{Selective influences, bijective form}]
\label{def:(SI,-bijective)} An arrow diagram (\ref{diag:bijective DSI})
is said to be the diagram of selective influences for $(A_{1},\ldots,A_{n})(\phi)$
and $(\alpha_{1},\ldots,\alpha_{n})$, and we write 
\[
(A_{1},\ldots,A_{n})\looparrowleft(\alpha_{1},\ldots,\alpha_{n}),
\]
if, for some random entity $R$ and for any treatment $\phi=\left\{ x_{1}^{\alpha_{1}},\ldots,x_{n}^{\alpha_{n}}\right\} \in T$,
\begin{equation}
\begin{array}{r}
(A_{1},\ldots,A_{n})(\phi)\sim\left(f_{1}(\phi_{\{\alpha_{1}\}},R),\ldots,f_{n}(\phi_{\{\alpha_{n}\}},R)\right)\\
\\
=\left(f_{1}(x_{1}^{\alpha_{1}},R),\ldots,f_{n}(x_{n}^{\alpha_{n}},R)\right),
\end{array}
\end{equation}
where $f_{i}:\alpha_{i}\times\mathcal{R}\rightarrow\mathcal{A}_{i}$
($i=1,\ldots,n$) are some functions, with $\mathcal{R}$ denoting
the set of possible values of $R$.%
\footnote{It will be shown below, Theorem \ref{thm:In-Definition-,}, that random
entity $R$ can always be chosen to be a random variable (in the narrow
sense).%
} 

\

This definition is difficult to put to work, as it refers to an existence
of a random entity (variable) $R$ without showing how one can find
it or prove that it cannot be found. The following criterion (necessary
and sufficient condition) for $(A_{1},\ldots,A_{n})\looparrowleft(\alpha_{1},\ldots,\alpha_{n})$
circumvents this problem. \end{defn}
\begin{criterion}[\emph{Joint Distribution Criterion}, JDC]
A vector of random variables $A(\phi)=(A_{1},\ldots,A_{n})(\phi)$
satisfies a diagram of selective influences (\ref{diag:bijective DSI})
if and only if there is a vector of jointly distributed random variables
\[
H=\left(\overset{\textnormal{for }\alpha_{1}}{\overbrace{H_{x_{1}^{\alpha_{1}}},\ldots,H_{x_{k_{1}}^{\alpha_{1}}}}},\ldots,\overset{\textnormal{for }\alpha_{n}}{\overbrace{H_{x_{1}^{\alpha_{n}}},\ldots,H_{x_{k_{n}}^{\alpha_{n}}}}}\right),
\]
one random variable for each factor point of each factor, such that
\begin{equation}
\left(H_{\phi_{\{\alpha_{1}\}}},\dots,H_{\phi_{\{\alpha_{n}\}}}\right)\sim A(\phi)
\end{equation}
for every treatment $\phi\in T$\textup{\emph{.}}
\end{criterion}
See Dzhafarov and Kujala (2010) for a proof. The vector $H$ in the
formulation of the JDC is referred to as the \emph{JDC-vector for}
$A(\phi)$, or the \emph{hypothetical JDC-vector for} $A(\phi)$,
if the existence of such a vector of jointly distributed variables
is in question.

The JDC prompts a simple justification for our definition of selective
influences. Let, for example, $(A,B,C)\looparrowleft(\alpha,\beta,\gamma)$,
with $\alpha=\left\{ 1^{\alpha},2^{\alpha}\right\} $, $\beta=\left\{ 1^{\beta},2^{\beta},3^{\beta}\right\} $,
$\gamma=\left\{ 1^{\gamma},2^{\gamma},3^{\gamma},4^{\gamma}\right\} $.
Consider all treatments $\phi$ in which the factor point of $\alpha$
is fixed, say, at $1^{\alpha}$. If $(A,B,C)\looparrowleft(\alpha,\beta,\gamma)$,
then in the vectors of random variables
\[
(A,B,C)\left(1^{\alpha},2^{\beta},1^{\gamma}\right),(A,B,C)\left(1^{\alpha},2^{\beta},3^{\gamma}\right),(A,B,C)\left(1^{\alpha},3^{\beta},1^{\gamma}\right)
\]
the marginal distribution of the variable $A$ is one and the same,
\[
A\left(1^{\alpha},2^{\beta},1^{\gamma}\right)\sim A\left(1^{\alpha},2^{\beta},3^{\gamma}\right)\sim A\left(1^{\alpha},3^{\beta},1^{\gamma}\right).
\]
But the intuition of selective influences requires more: that we can
denote this variable $A\left(1^{\alpha}\right)$ because it \emph{preserves
its identity} (and not just its distribution) no matter what other
variables it is paired with, $(B,C)\left(2^{\beta},1^{\gamma}\right)$,
$(B,C)\left(2^{\beta},3^{\gamma}\right)$, or $(B,C)\left(3^{\beta},1^{\gamma}\right)$.
Analogous statements hold for $A\left(2^{\alpha}\right)$, $B\left(2^{\beta}\right)$,
$B\left(3^{\beta}\right)$, $C\left(1^{\gamma}\right)$, etc. The
JDC formalizes the intuitive notion of variables ``preserving their
identity'' when entering in various combinations with each other:
there are jointly distributed random variables 
\[
H_{1^{\alpha}},H_{2^{\alpha}},H_{1^{\beta}},H_{2^{\beta}},H_{3^{\beta}},H_{1^{\gamma}},H_{2^{\gamma}},H_{3^{\gamma}},H_{4^{\gamma}}
\]
whose identity is defined by this joint distribution; when $H_{1^{\alpha}}$
is combined with random variables $H_{2^{\beta}}$ and $H_{1^{\gamma}}$,
it forms the triad $(H_{1^{\alpha}},H_{2^{\beta}},H_{1^{\gamma}})$
whose distribution is the same as that of $(A,B,C)\left(1^{\alpha},2^{\beta},1^{\gamma}\right)$;
when the same random variable $H_{1^{\alpha}}$ is combined with random
variables $H_{2^{\beta}}$ and $H_{3^{\gamma}}$, the triad $(H_{1^{\alpha}},H_{2^{\beta}},H_{3^{\gamma}})$
is distributed as $(A,B,C)\left(1^{\alpha},2^{\beta},3^{\gamma}\right)$;
and so on --- the key concept being that it is \emph{one and the same}
$H_{1^{\alpha}}$ which is being paired with other variables, as opposed
to different random variables $A\left(1^{\alpha},2^{\beta},1^{\gamma}\right),A\left(1^{\alpha},2^{\beta},3^{\gamma}\right),A\left(1^{\alpha},3^{\beta},1^{\gamma}\right)$
which are identically distributed. See Dzhafarov and Kujala (2010)
for a demonstration that the identity is not generally preserved if
all we know is marginal selectivity (as defined in Section \ref{sub:Three-basic-properties}).

The following is an important consequence of JDC.
\begin{thm}
\label{thm:In-Definition-,}In Definition \ref{def:(SI,-bijective)},
\textup{the random entity }$R$ can always be chosen to be a random
variable. Moreover, $R$ can be chosen arbitrarily, as any continuously
(atomlessly) distributed random variable, e.g., uniformly distributed
between 0 and 1. \end{thm}
\begin{proof}
The first statement follows from the fact that $R$ can be chosen
to coincide with the JDC-vector $H$, so that
\[
f_{i}(x^{\alpha_{i}},H)=H_{x^{\alpha_{i}}}^{\alpha_{i}},
\]
for $i=1,\ldots,n$, and $x^{\alpha_{i}}\in\alpha_{i}$. The JDC-vector
$H$ is a random variable. The second statement follows from Theorem
1 in Dzhafarov \& Gluhovsky, 2006, based on a general result for standard
Borel spaces (e.g., in Kechris, 1995, p. 116).
\end{proof}

\subsection{Three basic properties of selective influences\label{sub:Three-basic-properties}}

For completeness, we list three other fundamental consequences of
JDC (Dzhafarov \& Kujala, 2010).

\subsubsection{Nestedness.}

For any subset $\left\{ i_{1},\ldots,i_{k}\right\} $ of $\left\{ 1,\ldots,n\right\} $,
if $(A_{1},\ldots,A_{n})\looparrowleft(\alpha_{1},\ldots,\alpha_{n})$
then $(A_{i_{1}},\ldots,A_{i_{k}})\looparrowleft(\alpha_{i_{1}},\ldots,\alpha_{i_{k}})$.

\subsubsection{Complete Marginal Selectivity}

For any subset $\left\{ i_{1},\ldots,i_{k}\right\} $ of $\left\{ 1,\ldots,n\right\} $,
if $(A_{1},\ldots,A_{n})\looparrowleft(\alpha_{1},\ldots,\alpha_{n})$
then the $k$-marginal distribution%
\footnote{$k$-marginal distribution is the distribution of a subset of $k$
random variables ($k\geq1$) in a set of $n\geq k$ variables. In
Townsend and Schweickert (1989) the property was formulated for 1-marginals
of a pair of random variables. The adjective ``complete'' we use
with ``marginal selectivity'' is to emphasize that we deal with
all possible marginals rather than with just 1-marginals. %
} of $(A_{i_{1}},\ldots,A_{i_{k}})(\phi)$ does not depend on points
of the factors outside $(\alpha_{i_{1}},\ldots,\alpha_{i_{k}})$.
In particular, the distribution of $A_{i}$ only depends on points
of $\alpha_{i}$, $i=1,\ldots,n$. 

This is, of course, a trivial consequence of the nestedness property,
but its importance lies in that it provides the easiest to check necessary
condition for selective influences.

\subsubsection{\label{sub:Invariance}Invariance under factor-point-specific transformations}

Let $(A_{1},\ldots,A_{n})\looparrowleft(\alpha_{1},\ldots,\alpha_{n})$
and 
\[
H=\left(H_{x_{1}^{\alpha_{1}}},\ldots,H_{x_{k_{1}}^{\alpha_{i}}},\ldots,H_{x_{1}^{\alpha_{n}}},\ldots,H_{x_{k_{n}}^{\alpha_{n}}}\right)
\]
be the JDC-vector for $(A_{1},\ldots,A_{n})(\phi)$. Let $F$ be any
function that applies to $H$ componentwise and produces a corresponding
vector of random variables
\[
F\left(H\right)=\left(\begin{array}{c}
F\left(x_{1}^{\alpha_{1}},H_{x_{1}^{\alpha_{1}}}\right),\ldots,F\left(x_{k_{1}}^{\alpha_{i}},H_{x_{k_{1}}^{\alpha_{i}}}\right),\\
\ldots,\\
F\left(x_{1}^{\alpha_{n}},H_{x_{1}^{\alpha_{n}}}\right),\ldots,F\left(x_{k_{n}}^{\alpha_{n}},H_{x_{k_{n}}^{\alpha_{n}}}\right)
\end{array}\right),
\]
where we denote by $F\left(x^{\alpha},\cdot\right)$ the application
of $F$ to the component labeled by $x^{\alpha}$. Clearly, $F\left(H\right)$
possesses a joint distribution and contains one component for each
factor point. If we now define a vector of random variables $B\left(\phi\right)$
for every treatment $\phi\in T$ as
\[
(B_{1},\ldots,B_{n})\left(\phi\right)=\left(F\left(\phi_{\left\{ \alpha_{1}\right\} },A_{1}\right),\ldots,F\left(\phi_{\left\{ \alpha_{n}\right\} },A_{n}\right)\right)\left(\phi\right),
\]
then it follows from JDC that $(B_{1},\ldots,B_{n})\looparrowleft(\alpha_{1},\ldots,\alpha_{n})$.%
\footnote{Since it is possible that $F\left(x^{\alpha},H_{x^{\alpha}}\right)$
and $F\left(y^{\alpha},H_{y^{\alpha}}\right)$, with $x^{\alpha}\not=y^{\alpha}$,
have different sets of possible values, strictly speaking, one may
need to redefine the functions to ensure that the sets of possible
values for $B\left(\phi\right)$ is the same for different $\phi$.
This is, however, not essential (see footnote \ref{fn:Invariance}).%
} A function $F\left(x^{\alpha_{i}},\cdot\right)$ can be referred
to as a \emph{factor-point-specific transformation} of the random
variable $A_{i}$, because the random variable is transformed differently
for different points of the factor assumed to selectively influence
it. We can formulate the property in question by saying that a diagram
of selective influences is invariant under all factor-point-specific
transformations of the random variables. Note that this includes as
a special case transformations which are not factor-point-specific,
with 
\[
F\left(x_{1}^{\alpha_{i}},\cdot\right)\equiv\ldots\equiv F\left(x_{k_{i}}^{\alpha_{i}},\cdot\right)\equiv F\left(\alpha_{i},\cdot\right).
\]
This property is important for construction and use of tests for selective
influences (Dzhafarov \& Kujala, 2010; Kujala \& Dzhafarov, 2008).

\subsection{\label{sec:Quantum-Entanglement-and}Quantum entanglement and selective
influences}

In psychology, the notion of selective influences was introduced by
Sternberg (1969), in the context of studying ``stages'' of information
processing. Sternberg acknowledged that selective influences can hold
even if the durations of the stages being selectively affected are
not stochastically independent, but he lacked the mathematical apparatus
for dealing with this possibility. Townsend (1984) was the first to
study the notion of selectiveness under stochastic interdependence
systematically. He proposed to formalize the notion of selectively
influenced and stochastically interdependent random variables by the
concept of ``\emph{indirect nonselectiveness}'': the conditional
distribution of the variable $A$$_{1}$ given any value $a_{2}$
of the variable $A_{2}$, depends on $\alpha_{1}$ only, and, by symmetry,
the conditional distribution of $A_{2}$ at any $A_{1}=a_{1}$ depends
on $\alpha_{2}$ only. Under the name of ``\emph{conditionally selective
influence}'' this notion was mathematically characterized and generalized
in Dzhafarov (1999). It turned out, however, that this notion could
not serve as a general definition of selective influences, because
it did not satisfy some intuitive desiderata for such a definition,
e.g., the nestedness and marginal selectivity properties formulated
in Section \ref{sub:Three-basic-properties}. Variants of Definition
\ref{def:(SI,-bijective)} of the present paper were proposed in Dzhafarov
(2003) and both elaborated and generalized in Dzhafarov and Gluhovsky
(2006), Kujala and Dzhafarov (2008); JDC was explicitly formulated
in Dzhafarov and Kujala (2010), although clearly implied in the earlier
work. 

Until very recently (see footnote \ref{fn:We-are-grateful}) we were
blissfully unaware of the analogous developments in quantum physics.
The most conspicuous parallels can be found in Fine (1981a-b), but
that work in turn builds on a venerable line of research and thinking:
going back first to Bell (1964), and ultimately to Einstein, Podolsky,
and Rosen's (1935) paper. The issue in question regards two ``noncommuting''
measurements, such as those of the momentum and of the location of
a particle, or spin measurements along two different axes. For our
purposes it is sufficient to state that when one of two noncommuting
measurements is performed (without uncertainty about the result),
the second one cannot be performed on the same system. The key insight
needed to understand the analogy with the problem of selective influences
is this:\emph{ noncommuting measurements on the same system, being
mutually exclusive, can be viewed as levels (mutually exclusive values)
of one and the same external factor}. 

This is not entirely intuitive. Consider two particles for each of
which one can measure its momentum or its location. The analogy requires
that one view the measurement on particle 1 as a factor $\alpha_{1}$
with two mutually exclusive levels, $1^{\alpha_{1}}$ (location measurement)
and $2^{\alpha_{1}}$ (momentum measurement); and the measurement
on particle 2 is a factor $\alpha_{2}$ with two mutually exclusive
levels, $1^{\alpha_{2}}$ and $2^{\alpha_{2}}$, interpreted analogously.
The two measurements can be combined in treatments, $\left(1^{\alpha_{1}},1^{\alpha_{2}}\right)$,
$\left(1^{\alpha_{1}},2^{\alpha_{2}}\right)$, etc., but not within
a factor, $\left(1^{\alpha_{1}},2^{\alpha_{1}}\right)$ or $\left(1^{\alpha_{2}},2^{\alpha_{2}}\right)$.
The results of each of the measurements is a random variable, $A_{1}$
for particle 1 and $A_{2}$ for particle 2. The possible values $\mathcal{A}_{1}$
for $A_{1}$ are possible locations of particle 1 if $\alpha_{1}$
is at level $1^{\alpha_{1}}$, but they are possible momentum values
for particle 1 if $\alpha_{1}$ is at level $2^{\alpha_{1}}$ (which
makes it awkward but still possible to maintain the convention mentioned
in footnote \ref{fn:Invariance}). It is easier with spins (Bohm \&
Aharonov, 1957): for instance, for spin$\textnormal{-}\nicefrac{1}{2}$
particles (such as electrons), $\mathcal{A}_{1}$ consists of two
possible values of spin in one direction if $\alpha_{1}$ is at level
$1^{\alpha_{1}}$ and of two possible values of spin in another direction
if the level is $2^{\alpha_{1}}$. These two two-element sets are
more natural to consider ``the same.'' 

With all this in mind, the question now can be posed in the familiar
to us form: can we say that $\left(A_{1},A_{2}\right)\looparrowleft\left(\alpha_{1},\alpha_{2}\right)$,
or can the measurement (factor) $\alpha_{1}$ influence the result
(random variable) $A_{2}$ and/or $\alpha_{2}$ influence $A_{1}$?
In the Einstein-Podolsky-Rosen (EPR) paradigm involving \emph{entangled}
particles, the two random outcomes $A_{1},A_{2}$ are stochastically
interdependent, and their joint distribution at every treatment is
(correctly) predicted by the quantum theory. The question therefore
becomes: are the predicted (and observed) joint distributions of $\left(A_{1},A_{2}\right)$
compatible with the hypothesis $\left(A_{1},A_{2}\right)\looparrowleft\left(\alpha_{1},\alpha_{2}\right)$?
Einstein, Podolsky, and Rosen (1935) took $\left(A_{1},A_{2}\right)\looparrowleft\left(\alpha_{1},\alpha_{2}\right)$
for granted if the two particles are separated in space and measured
simultaneously (in some inertial frame of reference).

Bell's (1964) celebrated theorem shows that $\left(A_{1},A_{2}\right)\looparrowleft\left(\alpha_{1},\alpha_{2}\right)$
is not the case for entangled spin$\textnormal{-}\nicefrac{1}{2}$
particles obeying the laws of quantum mechanics. The reason this result
is considered to be of foundational importance (``the most profound
discovery in science,'' repeating the oft-quoted characterization
by Stapp, 1975) is that Bell essentially adopted Definition \ref{def:(SI,-bijective)}
for $\left(A_{1},A_{2}\right)\looparrowleft\left(\alpha_{1},\alpha_{2}\right)$
and identified the random entity $R$ with the set of all hidden variables
of a conceivable theory ``explaining'' the dependence of $\left(A_{1},A_{2}\right)$
on $\left(\alpha_{1},\alpha_{2}\right)$: knowing a value of $R$
one would be able to predict, through the functions $f_{1}$ and $f_{2}$
of Definition \ref{def:(SI,-bijective)}, the values of $\left(A_{1},A_{2}\right)$.
In addition to being called ``hidden'' the variables entailed in
$R$ are referred to as ``context-independent'' (meaning that the
distribution of $R$ and the functions $f_{1}$,$f_{2}$ do not depend
on treatments) and ``local'' (meaning, essentially, that in the
theory involving $R$ and $f_{1}$,$f_{2}$ the measurement $\alpha_{1}$
does not influence $A_{2}$, nor $\alpha_{2}$ influences $A_{1}$).
Bell's (1964) theorem therefore is interpreted as stating that quantum
predictions regarding two entangled spin$\textnormal{-}\nicefrac{1}{2}$
particles cannot be explained by any theory involving context-independent
and local variables. 

A rejection of $\left(A_{1},A_{2}\right)\looparrowleft\left(\alpha_{1},\alpha_{2}\right)$
in quantum physics can be handled by dispensing with locality (Bohm's
approach), but most physicists find this untenable (measurement $\alpha_{1}$
cannot influence $A_{2}$ if they are separated by a space-like interval).
The quantum probability theory can be viewed as a way of allowing
for context-dependence while retaining locality. In behavioral applications
both locality and context-independence can be targeted when $\left(A_{1},A_{2}\right)\looparrowleft\left(\alpha_{1},\alpha_{2}\right)$
is rejected, and distinguishing the two is a challenge.

Following the logic of Bell's work, Clauser, Horne, Shimony, \& Holt
(1969) derived a system of inequalities that are necessary conditions
for $\left(A_{1},A_{2}\right)\looparrowleft\left(\alpha_{1},\alpha_{2}\right)$
in the EPR paradigm with two particles and two measurements (factors)
with binary outcomes. These inequalities are subsumed in Fine's (1982a-b)
ones (discussed in Section \ref{sub:Quantum-entanglement}), which
present both necessary and sufficient conditions for $\left(A_{1},A_{2}\right)\looparrowleft\left(\alpha_{1},\alpha_{2}\right)$,
based on JDC. The latter was introduced in Fine's papers for the first
time (and called by this name too), although the earlier Suppes and
Zanotti's (1981) Theorem on Common Causes can also be viewed as a
special form of JDC. 

Fine's inequalities form a special case of the Linear Feasibility
Test considered in the next section. We therefore defer further discussion
of the EPR paradigm to Section \ref{sub:Quantum-entanglement}, and
conclude the present section by the following table of correspondences:

\begin{widetext}

\begin{center}
\begin{tabular}{|c|c|}
\hline 
{\small Selective Probabilistic Causality} & {\small Quantum Entanglement Problem (for spins)}\tabularnewline
\hline 
\hline 
{\footnotesize observed random output} & {\footnotesize detected spin value of a given particle}\tabularnewline
\hline 
{\footnotesize factor/input} & {\footnotesize spin measurement in a given particle}\tabularnewline
\hline 
{\footnotesize factor level} & {\footnotesize setting (axis) of the spin measurement}\tabularnewline
\hline 
{\footnotesize joint distribution criterion} & {\footnotesize joint distribution criterion}\tabularnewline
\hline 
{\footnotesize canonical diagram of selective influences} & {\footnotesize ``classical'' explanation (by context-independent
local variables)}\tabularnewline
\hline 
\end{tabular}
\par\end{center}

\end{widetext}

\section{\label{sec:Linear-Feasibility-Test}Linear Feasibility Test}

In this section we assume that for each random variable $A_{i}(\phi)$
in $(A_{1},\ldots,A_{n})(\phi)$ the set $\mathcal{A}_{i}$ of its
possible values has $m_{i}$ elements, $a_{1}^{i},\ldots,a_{m_{i}}^{i}$.
It is arguably the most important special case both because it is
ubiquitous and because in all other cases random variables can be
discretized into finite number of categories. We are interested in
establishing the truth or falsity of the diagram of selective influences
(\ref{diag:bijective DSI}), where each factor $\alpha_{j}$ in $(\alpha_{1},\ldots,\alpha_{n})$
contains $k_{j}$ factor points $x_{1}^{j},\ldots,x_{k_{j}}^{j}$
(written so instead of more formal $x_{1}^{\alpha_{j}},\ldots,x_{k_{j}}^{\alpha_{j}}$).
The \emph{Linear Feasibility Test} (LFT) to be described is a direct
application of JDC to this situation, furnishing a necessary and sufficient
condition for the diagram of selective influences $(A_{1},\ldots,A_{n})\looparrowleft(\alpha_{1},\ldots,\alpha_{n})$.

\subsection{The test}

In the hypothetical JDC-vector
\[
H=\left(H_{x_{1}^{1}},\ldots,H_{x_{k_{1}}^{1}},\ldots,H_{x_{1}^{n}},\ldots,H_{x_{k_{n}}^{n}}\right),
\]
since we assume that, for any point $x_{j}^{i}$ of factor $\alpha^{i}$
and any treatment $\phi$ containing $x_{j}^{i}$, 
\[
H_{x_{j}^{i}}\sim A_{i}\left(\phi\right),
\]
we know that the set of possible values for the random variable $H_{x_{j}^{i}}$
is $\left\{ a_{1}^{i},\ldots,a_{m_{i}}^{i}\right\} $, irrespective
of $j$. Denote 
\begin{equation}
\begin{array}{r}
\begin{array}{l}
\Pr\left[\left(A_{1}=a_{l_{1}}^{1},\ldots,A_{n}=a_{l_{n}}^{n}\right)\left(x_{j_{1}}^{1},\ldots,x_{j_{n}}^{n}\right)\right]\\
\\
=P\left(\stackrel{\textnormal{for r.v.s}}{\overbrace{l_{1},\ldots,l_{n}}}\,;\,\stackrel{\textnormal{for factor points}}{\overbrace{j_{1},\ldots,j_{n}}}\right),
\end{array}\end{array}\label{eq:p's}
\end{equation}
where $l_{i}\in\left\{ 1,\ldots,m_{i}\right\} $ and $j_{i}\in\left\{ 1,\ldots,k_{i}\right\} $
for $i=1,\ldots,n$ (``r.v.s'' abbreviates ``random variables'').
Denote 
\begin{equation}
\begin{array}{l}
\Pr\left[\begin{array}{c}
H_{x_{1}^{1}}=a_{l_{11}}^{1},\ldots,H_{x_{k_{1}}^{1}}=a_{l_{1k_{1}}}^{1},\\
\ldots,\\
H_{x_{1}^{n}}=a_{l_{n1}}^{n},\ldots,H_{x_{k_{n}}^{n}}=a_{l_{nk_{n}}}^{n}
\end{array}\right]\\
\\
=Q\left(\stackrel{\textnormal{for }A_{1}}{\overbrace{l_{11},\ldots,l_{1k_{1}}}},\ldots,\stackrel{\textnormal{for }A_{n}}{\overbrace{l_{n1},\ldots,l_{nk_{n}}}}\right),
\end{array}\label{eq:q's}
\end{equation}
where $l_{ij}\in\left\{ 1,\ldots,m_{i}\right\} $ for $i=1,\ldots,n$.
This gives us $m_{1}^{k_{1}}\times\ldots\times m_{n}^{k_{n}}$ $Q\textnormal{-}$probabilities.
A required joint distribution for the JDC-vector $H$ exists if and
only if these probabilities can be found subject to $m_{1}^{k_{1}}\times\ldots\times m_{n}^{k_{n}}$
nonnegativity constraints

\begin{equation}
Q\left(l_{11},\ldots,l_{1k_{1}},\ldots,l_{n1},\ldots,l_{nk_{n}}\right)\geq0,\label{eq:nonegativity}
\end{equation}
and (denoting by $n_{T}$ the number of treatments in $T$) $n_{T}\times m_{1}\times\ldots\times m_{n}$
linear equations
\begin{equation}
\begin{array}{r}
\sum Q\left(l_{11},\ldots,l_{1k_{1}},\ldots,l_{n1},\ldots,l_{nk_{n}}\right)\\
\\
=P\left(l_{1},\ldots,l_{n};j_{1},\ldots,j_{n}\right),
\end{array}\label{eq:linear equations}
\end{equation}
where the summation is across all possible values of the $\left(l_{11},\ldots,l_{1k_{1}},\ldots,l_{n1},\ldots,l_{nk_{n}}\right)$
subject to
\[
l_{1j_{1}}=l_{1},\ldots,l_{nj_{n}}=l_{n}.\footnotemark
\]

\footnotetext{The sum of all $Q$'s is 1 because it equals the sum of all $P$'s (across all $l_1,\ldots,l_n$) for any given treatment $j_1,\ldots,j_n$.}

This can be more compactly formulated in a matrix form. Let the observable
probabilities $P\left(l_{1},\ldots,l_{n};j_{1},\ldots,j_{n}\right)$
constitute components of a $n_{T}\times m_{1}\times\ldots\times m_{n}\textnormal{-}$dimensional
column vector $\mathbf{P}$, with its cells lexicographically enumerated
by $\left(l_{1},\ldots,l_{n};j_{1},\ldots,j_{n}\right)$. Let the
hypothetical probabilities $Q\left(l_{11},\ldots,l_{1k_{1}},\ldots,l_{n1},\ldots,l_{nk_{n}}\right)$
constitute components of a $m_{1}^{k_{1}}\times\ldots\times m_{n}^{k_{n}}\textnormal{-}$dimensional
column vector $\mathbf{Q}$, with its cells lexicographically enumerated
by $\left(l_{11},\ldots,l_{1k_{1}},\ldots,l_{n1},\ldots,l_{nk_{n}}\right)$.
Let $\mathbf{M}$ be a Boolean matrix with $n_{T}\times m_{1}\times\ldots\times m_{n}$
rows and $m_{1}^{k_{1}}\times\ldots\times m_{n}^{k_{n}}$ columns
lexicographically enumerated in the same way as, respectively, $\mathbf{P}$
and $\mathbf{Q}$, such that the entry in the cell in the $\left(l_{1},\ldots,l_{n};j_{1},\ldots,j_{n}\right)$th
row and $\left(l_{11},\ldots,l_{1k_{1}},\ldots,l_{n1},\ldots,l_{nk_{n}}\right)$th
column is 1 if $l_{1j_{1}}=l_{1},\ldots,l_{nj_{n}}=l_{n}$; otherwise
the entry is 0. Clearly, the vector $Q$ exists if and only if the
system 
\begin{equation}
\mbox{\ensuremath{\mathbf{MQ}=\mathbf{P},\;\mathbf{Q}\geq0}}\label{eq:mq=00003Dp}
\end{equation}
(with the inequality understood componentwise) has a solution. This
is a typical \emph{linear programming} (LP) problem. More precisely,
this is an LP task in the standard form and with a dummy objective
function (e.g., a linear combination with zero coefficients). It is
known (Karmarkar, 1984; Khachiyan, 1979) that it is always possible,
in polynomial time, to either find a solution for such a system or
to determine that it does not exist. Many standard software packages
can handle this problem (e.g., GNU Linear Programming Kit at http://www.gnu.org/software/glpk/).

\subsection{\label{sub:Properties-of-the}Properties of the LP problem}

The rank of matrix $\mathbf{M}$ is always strictly smaller than the
number of components in $\mathbf{P}$. This follows from the fact
that for any two allowable treatments $\left(j_{1},\ldots,j_{n}\right)$
and $\left(j'_{1},\ldots,j'_{n}\right)$ that share a subvector 
\[
\left(j_{1'},\ldots,j_{s'}\right)=\left(j'_{1'},\ldots,j'_{s'}\right)
\]
(where we use $\left\{ 1',\ldots,s'\right\} $ to designate $s$ distinct
elements chosen from $\left\{ 1,\ldots,n\right\} $), and for any
fixed $\left(v_{1},\ldots,v_{s}\right)$, the sum of all rows of $\mathbf{M}$
corresponding to $\left(l_{1},\ldots,l_{n};j_{1},\ldots,j_{n}\right)$th
components of $\mathbf{P}$ with $\left(l_{1'},\ldots,l_{s'}\right)=\left(v_{1},\ldots,v_{s}\right)$
is the same Boolean vector as the sum of all rows of $\mathbf{M}$
corresponding to $\left(l_{1},\ldots,l_{n};j'_{1},\ldots,j'_{n}\right)$th
components of $\mathbf{P}$ with the same property. The upper limit
for the rank of matrix $\mathbf{M}$ is given in the following theorem. 
\begin{thm}
\label{thm:The-rank-of}The rank of $\mathbf{M}$ for a maximal set
of treatments $T=\alpha_{1}\times\ldots\times\alpha_{n}$ is \textup{
\[
\left(k_{1}\left(m_{1}-1\right)+1\right)\ldots\left(k_{n}\left(m_{n}-1\right)+1\right).
\]
}\end{thm}
\begin{proof}
Given any
\[
\begin{array}{c}
\left\{ 1',\ldots,s'\right\} \subset\left\{ 1,\ldots,n\right\} ,\\
\left(j_{1'},\ldots,j_{s'}\right)\in\left\{ 1,\ldots,k_{1'}\right\} \times\ldots\times\left\{ 1,\ldots,k_{s'}\right\} ,\\
\left(l_{1'},\ldots,l_{s'}\right)\in\left\{ 1,\ldots,m_{1'}\right\} \times\ldots\times\left\{ 1,\ldots,m_{s'}\right\} ,
\end{array}
\]
let $\mathbf{V}\left(1',\ldots,s';j_{1'},\ldots,j_{s'};l_{1'},\ldots,l_{s'}\right)$
denote an $\left(m_{1}\right)^{k_{1}}\ldots\left(m_{n}\right)^{k_{n}}\textnormal{-}$component
Boolean row vector whose components are lexicographically enumerated
in the same way as $\mathbf{Q}$, and such that its $\left(l_{11},\ldots,l_{1k_{1}},\ldots,l_{n1},\ldots,l_{nk_{n}}\right)$th
component is 1 if and only if 
\[
l_{1'j_{1'}}=l_{1'},\ldots,l_{s'j_{s'}}=l_{s'}.
\]
The rows of matrix $\mathbf{M}$ are $\mathbf{V}\left(1,\ldots,n;j_{1},\ldots,j_{n};l_{1},\ldots,l_{n}\right)$-vectors.
It is easy to check that for any fixed $\left(1',\ldots,s';j_{1'},\ldots,j_{s'}\right)$,
the sum of the rows of $\mathbf{M}$ corresponding to fixed values
$\left(l_{1'},\ldots,l_{s'}\right)$ is $\mathbf{V}\left(1',\ldots,s';j_{1'},\ldots,j_{s'};l_{1'},\ldots,l_{s'}\right)$.
It follows that for $s=n,n-1,\ldots,1$, a vector $\mathbf{V}\left(1',\ldots,s';j_{1'},\ldots,j_{s'};l_{1'},\ldots,l_{s'}\right)$
in which all $l_{i'}=1$ except for $i'\in\left\{ 1'',\ldots,v''\right\} \subset\left\{ 1',\ldots,s'\right\} $
(a subset of $v<s$ distinct elements), is a linear combination of
the vector 
\[
\mathbf{V}\left(1'',\ldots,v'';j_{1''},\ldots,j_{v''};l_{1''},\ldots,l_{v''}\right)
\]
and all the vectors 
\[
\mathbf{V}\left(1',\ldots,s';j_{1'},\ldots,j_{s'};l_{1'},\ldots,l_{s'}\right)
\]
 for which all $l_{i'}>1$ and 
\[
\left\{ j_{1''},\ldots,j_{v''};l_{1''},\ldots,l_{v''}\right\} \subset\left\{ j_{1'},\ldots,j_{s'};l_{1'},\ldots,l_{s'}\right\} .
\]
As a result the rows of $\mathbf{M}$ are linear combinations of the
rows of $\mathbf{M}^{*}$ consisting of vectors 
\[
\mathbf{V}\left(1',\ldots,s';j_{1'},\ldots,j_{s'};l_{1'},\ldots,l_{s'}\right)
\]
for all possible
\[
\begin{array}{c}
\left\{ 1',\ldots,s'\right\} \subset\left\{ 1,\ldots,n\right\} ,\\
\left(j_{1'},\ldots,j_{s'}\right)\in\left\{ 1,\ldots,k_{1'}\right\} \times\ldots\times\left\{ 1,\ldots,k_{s'}\right\} ,\\
\left(l_{1'},\ldots,l_{s'}\right)\in\left\{ 2,\ldots,m_{1'}\right\} \times\ldots\times\left\{ 2,\ldots,m_{s'}\right\} .
\end{array}
\]
By straightforward combinatorics the number of such vectors is 
\[
\left(k_{1}\left(m_{1}-1\right)+1\right)\ldots\left(k_{n}\left(m_{n}-1\right)+1\right).
\]
The rows of $\mathbf{M}^{*}$ are linearly independent because the
column corresponding to the $\left(l_{11}=1,\ldots,l_{1k_{1}}=1,\ldots,l_{n1}=1,\ldots,l_{nk_{n}}=1\right)$th
component of $Q$ contains a single 1, in the row of $\mathbf{M}^{*}$
corresponding to $s=0$ (which row contains 1's only).
\end{proof}
Note that 
\[
k_{i}\left(m_{i}-1\right)+1<m_{i}^{k_{i}}
\]
for all $k_{i}\ge2$ and $m_{i}\ge1$. This means that 
\[
\left(k_{1}\left(m_{1}-1\right)+1\right)\ldots\left(k_{n}\left(m_{n}-1\right)+1\right)<\left(m_{1}\right)^{k_{1}}\ldots\left(m_{n}\right)^{k_{n}},
\]
and the system $\mathbf{MQ=P}$ is always underdetermined.
\begin{cor}
If $\mathbf{P}$ satisfies marginal selectivity, then system (\ref{eq:mq=00003Dp})
is equivalent to 
\begin{equation}
\mathbf{M}^{*}\mathbf{Q}=\mathbf{P}^{*},\;\mathbf{Q}\geq0,\label{eq:m*q=00003Dp*}
\end{equation}
where $\mathbf{M}^{*}$ is as defined in the proof above, and $\mathbf{P}^{*}$
is the ``reduced hierarchical'' vector with components 
\begin{equation}
\begin{array}{r}
\Pr\left[\left(A_{1'}=a_{l_{1'}}^{1'},\ldots,A_{s'}=a_{l_{s'}}^{s'}\right)\left(x_{j_{1'}}^{1'},\ldots,x_{j_{s'}}^{s'}\right)\right]\\
=P_{1',\ldots,s'}^{*}\left(l_{1'},\ldots,l_{s'};j_{1'},\ldots,j_{s'}\right),
\end{array}\label{eq:p*'s}
\end{equation}
where $s=0,\ldots,n,$ \textup{$\left\{ 1',\ldots,s'\right\} \subset\left\{ 1,\ldots,n\right\} $},
and $l_{i'}\in$$\left\{ 2,\ldots,m_{i}\right\} $ for each $i'\in\left\{ 1',\ldots,s'\right\} $.
$\mathbf{M}^{*}$ is of full row rank.
\end{cor}
To comment on this corollary, it follows from the proof of Theorem
\ref{thm:The-rank-of} that $\mathbf{MQ}=\mathbf{P}$ never has a
solution if vector $\mathbf{P}$ violates the equality
\[
\begin{array}{l}
\sum\Pr\left[\left(A_{1}=a_{l_{1}}^{1},\ldots,A_{n}=a_{l_{n}}^{n}\right)\left(x_{j_{1}}^{1},\ldots,x_{j_{n}}^{n}\right)\right]\\
=\sum\Pr\left[\left(A_{1}=a_{l_{1}}^{1},\ldots,A_{n}=a_{l_{n}}^{n}\right)\left(x_{j'_{1}}^{1},\ldots,x_{j'_{n}}^{n}\right)\right],
\end{array}
\]
where the summation is across all values of $\left(l_{1},\ldots,l_{n}\right)$
with a fixed $\left(l_{1'},\ldots,l_{s'}\right)$. Clearly, this necessary
condition is just another way of stating marginal selectivity. Assuming
that $\mathbf{P}$ does satisfy marginal selectivity, it can be represented
by the ``reduced hierarchical'' vector $\mathbf{P}^{*}$ whose components
are marginal probabilities of all orders, with $s=0$ corresponding
to the probability 1.

\subsection{\label{sub:Examples}Examples}
\begin{example}
Let $\alpha=\left\{ 1^{\alpha},2^{\alpha}\right\} $, $\beta=\left\{ 1^{\beta},2^{\beta}\right\} $,
and the set of allowable treatments $T$ consist of all four possible
combinations of the factor points. Let $A$ and $B$ be Bernoulli
variables, $a_{1}=b_{1}=1$, $a_{2}=b_{2}=2$, distributed as shown:

\

\begin{center}%
\begin{tabular}{cc|cc|c|}
\hline 
\multicolumn{1}{|c}{$\alpha$} & $\beta$ & $A$ & $B$ & $\Pr$\tabularnewline
\hline 
\hline 
\multicolumn{1}{|c}{1} & 1 & 1 & 1 & $.140$\tabularnewline
\cline{1-2} 
 &  & 1 & 2 & $.360$\tabularnewline
 &  & 2 & 1 & $.360$\tabularnewline
 &  & 2 & 2 & $.140$\tabularnewline
\cline{3-5} 
\end{tabular}$\quad$%
\begin{tabular}{cc|cc|c|}
\hline 
\multicolumn{1}{|c}{$\alpha$} & $\beta$ & $A$ & $B$ & $\Pr$\tabularnewline
\hline 
\hline 
\multicolumn{1}{|c}{1} & 2 & 1 & 1 & $.198$\tabularnewline
\cline{1-2} 
 &  & 1 & 2 & $.302$\tabularnewline
 &  & 2 & 1 & $.302$\tabularnewline
 &  & 2 & 2 & $.198$\tabularnewline
\cline{3-5} 
\end{tabular}

\

\begin{tabular}{cc|cc|c|}
\hline 
\multicolumn{1}{|c}{$\alpha$} & $\beta$ & $A$ & $B$ & $\Pr$\tabularnewline
\hline 
\hline 
\multicolumn{1}{|c}{2} & 1 & 1 & 1 & $.189$\tabularnewline
\cline{1-2} 
 &  & 1 & 2 & $.311$\tabularnewline
 &  & 2 & 1 & $.311$\tabularnewline
 &  & 2 & 2 & $.189$\tabularnewline
\cline{3-5} 
\end{tabular}$\quad$%
\begin{tabular}{cc|cc|c|}
\hline 
\multicolumn{1}{|c}{$\alpha$} & $\beta$ & $A$ & $B$ & $\Pr$\tabularnewline
\hline 
\hline 
\multicolumn{1}{|c}{2} & 2 & 1 & 1 & $.460$\tabularnewline
\cline{1-2} 
 &  & 1 & 2 & $.040$\tabularnewline
 &  & 2 & 1 & $.040$\tabularnewline
 &  & 2 & 2 & $.460$\tabularnewline
\cline{3-5} 
\end{tabular}\end{center}

\

\protect{\noindent}Marginal selectivity here is satisfied trivially:
all marginal probabilities are equal 0.5, for all treatments. In the
matrix form of the LFT, the column-vector of the above 16 probabilities,
\[
(.140,.360,.360,\ldots,.040,.040,.460)^{\top},
\]
using $\top$ for transposition, is denoted by $\mathbf{P}$. The
LFT problem is defined by the system $\mathbf{M}\mathbf{Q=\mathbf{P}},$
$\mathbf{Q}\ge0,$ where the $16\times16$ Boolean matrix $\mathbf{M}$
is shown below: each column of the matrix corresponds to a combination
of values for the hypothetical $H\textnormal{-}$variables (shown
above the matrix), while each row corresponds to a combination of
a treatment with values of the outputs $A,B$ (shown on the left).

\[%
\begin{tabular}{cc|cc|c|cccccccccccccccc|}
\cline{4-21} 
 & \multicolumn{1}{c}{} & \multicolumn{1}{c|}{} & \multicolumn{2}{c|}{$H_{1^{\alpha}}$} & 1 & 1 & 1 & 1 & 1 & 1 & 1 & 1 & 2 & 2 & 2 & 2 & 2 & 2 & 2 & 2\tabularnewline
 & \multicolumn{1}{c}{} & \multicolumn{1}{c|}{} & \multicolumn{2}{c|}{$H_{2^{\alpha}}$} & 1 & 1 & 1 & 1 & 2 & 2 & 2 & 2 & 1 & 1 & 1 & 1 & 2 & 2 & 2 & 2\tabularnewline
 & \multicolumn{1}{c}{} & \multicolumn{1}{c|}{} & \multicolumn{2}{c|}{$H_{1^{\beta}}$} & 1 & 1 & 2 & 2 & 1 & 1 & 2 & 2 & 1 & 1 & 2 & 2 & 1 & 1 & 2 & 2\tabularnewline
 & \multicolumn{1}{c}{} & \multicolumn{1}{c|}{} & \multicolumn{2}{c|}{$H_{2^{\beta}}$} & 1 & 2 & 1 & 2 & 1 & 2 & 1 & 2 & 1 & 2 & 1 & 2 & 1 & 2 & 1 & 2\tabularnewline
\hline 
\multicolumn{1}{|c}{$\alpha$} & $\beta$ & $A$ & $B$ & \multicolumn{1}{c}{} &  &  &  &  &  &  &  &  &  &  &  &  &  &  &  & \multicolumn{1}{c}{}\tabularnewline
\cline{1-4} \cline{6-21} 
\multicolumn{1}{|c}{1} & 1 & 1 & 1 &  & 1 & 1 & $0$ & $0$ & 1 & 1 & $0$ & $0$ & $0$ & $0$ & $0$ & $0$ & $0$ & $0$ & $0$ & $0$\tabularnewline
\cline{1-2} 
 &  & 1 & 2 &  & $0$ & $0$ & 1 & 1 & $0$ & $0$ & 1 & 1 & $0$ & $0$ & $0$ & $0$ & $0$ & $0$ & $0$ & $0$\tabularnewline
 &  & 2 & 1 &  & $0$ & $0$ & $0$ & $0$ & $0$ & $0$ & $0$ & $0$ & 1 & 1 & $0$ & $0$ & 1 & 1 & $0$ & $0$\tabularnewline
 &  & 2 & 2 &  & $0$ & $0$ & $0$ & $0$ & $0$ & $0$ & $0$ & $0$ & $0$ & $0$ & 1 & 1 & $0$ & $0$ & 1 & 1\tabularnewline
\cline{1-4} 
\multicolumn{1}{|c}{1} & 2 & 1 & 1 &  & 1 & $0$ & 1 & $0$ & 1 & $0$ & 1 & $0$ & $0$ & $0$ & $0$ & $0$ & $0$ & $0$ & $0$ & $0$\tabularnewline
\cline{1-2} 
 &  & 1 & 2 &  & $0$ & 1 & $0$ & 1 & $0$ & 1 & $0$ & 1 & $0$ & $0$ & $0$ & $0$ & $0$ & $0$ & $0$ & $0$\tabularnewline
 &  & 2 & 1 &  & $0$ & $0$ & $0$ & $0$ & $0$ & $0$ & $0$ & $0$ & 1 & $0$ & 1 & $0$ & 1 & $0$ & 1 & $0$\tabularnewline
 &  & 2 & 2 &  & $0$ & $0$ & $0$ & $0$ & $0$ & $0$ & $0$ & $0$ & $0$ & 1 & $0$ & 1 & $0$ & 1 & $0$ & 1\tabularnewline
\cline{1-4} 
\multicolumn{1}{|c}{2} & 1 & 1 & 1 &  & 1 & 1 & $0$ & $0$ & $0$ & $0$ & $0$ & $0$ & 1 & 1 & $0$ & $0$ & $0$ & $0$ & $0$ & $0$\tabularnewline
\cline{1-2} 
 &  & 1 & 2 &  & $0$ & $0$ & 1 & 1 & $0$ & $0$ & $0$ & $0$ & $0$ & $0$ & 1 & 1 & $0$ & $0$ & $0$ & $0$\tabularnewline
 &  & 2 & 1 &  & $0$ & $0$ & $0$ & $0$ & 1 & 1 & $0$ & $0$ & $0$ & $0$ & $0$ & $0$ & 1 & 1 & $0$ & $0$\tabularnewline
 &  & 2 & 2 &  & $0$ & $0$ & $0$ & $0$ & $0$ & $0$ & 1 & 1 & $0$ & $0$ & $0$ & $0$ & $0$ & $0$ & 1 & 1\tabularnewline
\cline{1-4} 
\multicolumn{1}{|c}{2} & 2 & 1 & 1 &  & 1 & $0$ & 1 & $0$ & $0$ & $0$ & $0$ & $0$ & 1 & $0$ & 1 & $0$ & $0$ & $0$ & $0$ & $0$\tabularnewline
\cline{1-2} 
 &  & 1 & 2 &  & $0$ & 1 & $0$ & 1 & $0$ & $0$ & $0$ & $0$ & $0$ & 1 & $0$ & 1 & $0$ & $0$ & $0$ & $0$\tabularnewline
 &  & 2 & 1 &  & $0$ & $0$ & $0$ & $0$ & 1 & $0$ & 1 & $0$ & $0$ & $0$ & $0$ & $0$ & 1 & $0$ & 1 & $0$\tabularnewline
 &  & 2 & 2 &  & $0$ & $0$ & $0$ & $0$ & $0$ & 1 & $0$ & 1 & $0$ & $0$ & $0$ & $0$ & $0$ & 1 & $0$ & 1\tabularnewline
\cline{3-4} \cline{6-21} 
\end{tabular} \]

The linear programing routine of Mathematica\texttrademark (using
the interior point algorithm) shows that the linear equations (\ref{eq:linear equations})
have nonnegative solutions corresponding to the JDC-vector

\

\protect{\noindent}%
\begin{tabular}{|cccc|c|}
\hline 
$H_{1^{\alpha}}$ & $H_{2^{\alpha}}$ & $H_{1^{\beta}}$ & $H_{2^{\beta}}$ & $\Pr$\tabularnewline
\hline 
\hline 
1 & 1 & 1 & 1 & $.02708610$\tabularnewline
1 & 1 & 1 & 2 & $.00239295$\tabularnewline
1 & 1 & 2 & 1 & $.16689300$\tabularnewline
1 & 1 & 2 & 2 & $.03358610$\tabularnewline
1 & 2 & 1 & 1 & $.00197965$\tabularnewline
1 & 2 & 1 & 2 & $.10854100$\tabularnewline
1 & 2 & 2 & 1 & $.00204128$\tabularnewline
1 & 2 & 2 & 2 & $.15748000$\tabularnewline
\hline 
\end{tabular}$\quad$%
\begin{tabular}{|cccc|c|}
\hline 
$H_{1^{\alpha}}$ & $H_{2^{\alpha}}$ & $H_{1^{\beta}}$ & $H_{2^{\beta}}$ & $\Pr$\tabularnewline
\hline 
\hline 
2 & 1 & 1 & 1 & $.15748000$\tabularnewline
2 & 1 & 1 & 2 & $.00204128$\tabularnewline
2 & 1 & 2 & 1 & $.10854100$\tabularnewline
2 & 1 & 2 & 2 & $.00197965$\tabularnewline
2 & 2 & 1 & 1 & $.03358610$\tabularnewline
2 & 2 & 1 & 2 & $.16689300$\tabularnewline
2 & 2 & 2 & 1 & $.00239295$\tabularnewline
2 & 2 & 2 & 2 & $.02708610$\tabularnewline
\hline 
\end{tabular}

\

\protect{\noindent}The column-vector of these probabilities constitutes
$\mathbf{Q}>0$. This proves that in this case we do have $(A,B)\looparrowleft(\alpha,\beta)$.\qed
\end{example}

\begin{example}
In the previous example, let us change the distributions of $(A,B)$
to the following:

\

\begin{center}%
\begin{tabular}{cc|cc|c|}
\hline 
\multicolumn{1}{|c}{$\alpha$} & $\beta$ & $A$ & $B$ & $\Pr$\tabularnewline
\hline 
\hline 
\multicolumn{1}{|c}{1} & 1 & 1 & 1 & $.450$\tabularnewline
\cline{1-2} 
 &  & 1 & 2 & $.050$\tabularnewline
 &  & 2 & 1 & $.050$\tabularnewline
 &  & 2 & 2 & $.450$\tabularnewline
\cline{3-5} 
\end{tabular}$\quad$%
\begin{tabular}{cc|cc|c|}
\hline 
\multicolumn{1}{|c}{$\alpha$} & $\beta$ & $A$ & $B$ & $\Pr$\tabularnewline
\hline 
\hline 
\multicolumn{1}{|c}{1} & 2 & 1 & 1 & $.105$\tabularnewline
\cline{1-2} 
 &  & 1 & 2 & $.395$\tabularnewline
 &  & 2 & 1 & $.395$\tabularnewline
 &  & 2 & 2 & $.105$\tabularnewline
\cline{3-5} 
\end{tabular}

\

\begin{tabular}{cc|cc|c|}
\hline 
\multicolumn{1}{|c}{$\alpha$} & $\beta$ & $A$ & $B$ & $\Pr$\tabularnewline
\hline 
\hline 
\multicolumn{1}{|c}{2} & 1 & 1 & 1 & $.170$\tabularnewline
\cline{1-2} 
 &  & 1 & 2 & $.330$\tabularnewline
 &  & 2 & 1 & $.330$\tabularnewline
 &  & 2 & 2 & $.170$\tabularnewline
\cline{3-5} 
\end{tabular}$\quad$%
\begin{tabular}{cc|cc|c|}
\hline 
\multicolumn{1}{|c}{$\alpha$} & $\beta$ & $A$ & $B$ & $\Pr$\tabularnewline
\hline 
\hline 
\multicolumn{1}{|c}{2} & 2 & 1 & 1 & $.110$\tabularnewline
\cline{1-2} 
 &  & 1 & 2 & $.390$\tabularnewline
 &  & 2 & 1 & $.390$\tabularnewline
 &  & 2 & 2 & $.110$\tabularnewline
\cline{3-5} 
\end{tabular}\end{center}

\

\protect{\noindent}Once again, marginal selectivity is satisfied
trivially, as all marginal probabilities are 0.5, for all treatments.
The linear programing routine of Mathematica\texttrademark, however,
shows that the linear equations (\ref{eq:linear equations}) have
no nonnegative solutions. This excludes the existence of a JDC-vector
for this situations, ruling out thereby the possibility of $(A,B)\looparrowleft(\alpha,\beta)$.\qed
\end{example}

\subsection{Renaming and grouping}

Since LFT is both a necessary and sufficient condition for selective
influences, if it is passed for $(A_{1},\ldots,A_{n})(\phi)$, it
is guaranteed to be passed following any factor-point-specific transformations
of these random outputs. All such transformations in the case of discrete
random variables can be described as combinations of renaming (factor-point
specific one) and coarsening (grouping of some values together). In
fact, the outcome of LFT simply does not depend on the values of the
random variables involved, only their probabilities matter. Therefore
a renaming will not change anything in the system of linear equations
and inequalities (\ref{eq:nonegativity})-(\ref{eq:linear equations}).
An example of coarsening will be redefining $A$ and $B$, each having
possible values $1,2,3,4$, into binary variables
\[
A^{*}\left(\phi\right)=\begin{cases}
1 & \textnormal{if }A\left(\phi\right)=1,2,\\
2 & \textnormal{if }A\left(\phi\right)=3,4,
\end{cases}\quad B^{*}\left(\phi\right)=\begin{cases}
1 & \textnormal{if }B\left(\phi\right)=1,2,3,\\
2 & \textnormal{if }B\left(\phi\right)=4.
\end{cases}
\]
It is clear that any such a redefinition amounts to replacing some
of the equations in (\ref{eq:linear equations}) with their sums.
Therefore, if the original system has a solution, so will also the
system after such replacements. Of course, the reverse is not generally
true: the coarser system can have solutions when the original system
does not.

The same is true for coarsening the system by grouping together some
of the factor points within factors. Suppose we want to group together
points $x_{1}^{1}$ and $x_{2}^{1}$ of factor $\alpha_{1}$ containing
more than two points. This means that the probabilities $P\left(l_{1},l_{2},\ldots,l_{n};j_{1},j_{2},\ldots,j_{n}\right)$
are redefined as%
\footnote{More general mixtures, $\pi P\left(l_{1},l_{2},\ldots,l_{n};1,j_{2},\ldots,j_{n}\right)+\left(1-\pi\right)P\left(l_{1},l_{2},\ldots,l_{n};2,j_{2},\ldots,j_{n}\right)$
for $0<\pi\leq1$, are dealt with as easily; moreover, $\pi=1$ formally
corresponds to dropping the factor point $x_{2}^{1}$, considered
below. The values of $\pi$ other than $\nicefrac{1}{2}$ and 1 can
be useful if the grouping is done on a sample level, to reflect the
differences in sample sizes corresponding to treatments containing
$x_{1}^{1}$ and $x_{2}^{1}$. %
}
\[
\begin{array}{l}
P'\left(l_{1},l_{2},\ldots,l_{n};j_{1},j_{2},\ldots,j_{n}\right)\\
=\left\{ \begin{array}{l}
\frac{1}{2}P\left(l_{1},l_{2},\ldots,l_{n};1,j_{2},\ldots,j_{n}\right)+\frac{1}{2}P\left(l_{1},l_{2},\ldots,l_{n};2,j_{2},\ldots,j_{n}\right)\\
\textnormal{if }j_{1}=1,\\
P\left(l_{1},l_{2},\ldots,l_{n};j_{1}+1,j_{2},\ldots,j_{n}\right)\\
\textnormal{if }j_{1}>1.
\end{array}\right.
\end{array}
\]
 When we average the original equations for $P\left(l_{1},l_{2},\ldots,l_{n};1,j_{2},\ldots,j_{n}\right)$
and $P\left(l_{1},l_{2},\ldots,l_{n};2,j_{2},\ldots,j_{n}\right)$,
we get
\[
\begin{array}{r}
\sum\left\{ \begin{array}{c}
\frac{1}{2}\sum_{l_{12}}Q\left(l_{11}=l_{1},l_{12},\ldots,l_{1k_{1}},\ldots,l_{n1},\ldots,l_{nk_{n}}\right)\\
+\frac{1}{2}\sum_{l_{11}}Q\left(l_{11},l_{12}=l_{1}\ldots,l_{1k_{1}},\ldots,l_{n1},\ldots,l_{nk_{n}}\right)
\end{array}\right\} \\
=P'\left(l_{1},l_{2},\ldots,l_{n};1,j_{2},\ldots,j_{n}\right),
\end{array}
\]
where $l_{2j_{2}}=l_{2},\ldots,l_{nj_{n}}=l_{n}$ and the outer summation
is across all $l_{ij}$ except for the following values for $\left(i,j\right)$:
$\left(1,1\right)$, $\left(1,2\right)$, and $\left(i,j_{i}\right)$,
$i=2,\ldots,n$. We define a new vector $Q'$ whose dimensionality
is less than that of $Q$ by one, putting 
\[
\begin{array}{l}
Q'\left(l_{11}=l,l_{13},\ldots,l_{1k_{1}},\ldots,l_{n1},\ldots,l_{nk_{n}}\right)\\
=\frac{1}{2}\sum_{l_{12}}Q\left(l_{11}=l,l_{12},l_{13},\ldots,l_{1k_{1}},\ldots,l_{n1},\ldots,l_{nk_{n}}\right)\\
+\frac{1}{2}\sum_{l_{11}}Q\left(l_{11},l_{12}=l,l_{13},\ldots,l_{1k_{1}},\ldots,l_{n1},\ldots,l_{nk_{n}}\right),
\end{array}
\]
where $l$ has the same range as any of $l_{1j}$. (For notational
simplicity, in $Q'$ we do not re-enumerate $\left(1,3\right)$ as
$\left(1,2\right)$, $\left(1,4\right)$ as $\left(1,3\right)$, etc.,
leaving thereby $l_{12}$ undefined) 

For any point of factor $\alpha_{1}$ other than $x_{1}^{1}$ and
$x_{2}^{1}$, say, $x_{3}^{1}$, we have then
\[
\begin{array}{r}
\sum_{l_{11},l_{12}}Q\left(l_{11},l_{12},\ldots,l_{1k_{1}},\ldots,l_{n1},\ldots,l_{nk_{n}}\right)\\
=P\left(l_{1},l_{2}\ldots,l_{n};3,j_{2}\ldots,j_{n}\right),
\end{array}
\]
which can be presented as
\[
\begin{array}{r}
\sum\sum_{l}\left\{ \begin{array}{l}
\frac{1}{2}\sum_{l_{12}}Q\left(l_{11}=l,l_{12},l_{13}=l_{1},\ldots,l_{1k_{1}},\ldots,l_{n1},\ldots,l_{nk_{n}}\right)\\
+\frac{1}{2}\sum_{l_{11}}Q\left(l_{11},l_{12}=l,l_{13}=l_{1},\ldots,l_{1k_{1}},\ldots,l_{n1},\ldots,l_{nk_{n}}\right)
\end{array}\right\} \\
=P\left(l_{1},l_{2}\ldots,l_{n};3,j_{2}\ldots,j_{n}\right).
\end{array}
\]
This is equivalent to

\[
\begin{array}{r}
\sum Q'\left(l_{11},l_{13}=l_{1},\ldots,l_{1k_{1}},\ldots,l_{n1},\ldots,l_{nk_{n}}\right)\\
=P'\left(l_{1},l_{2}\ldots,l_{n};j_{1}=2,j_{2}\ldots,j_{n}\right),
\end{array}
\]
where $l_{2j_{2}}=l_{2},\ldots,l_{nj_{n}}=l_{n}$, and the summation
is across all $l_{ij}$ except for $\left(i,j\right)=\left(1,3\right)$
and $\left(i,j\right)=\left(i,j_{i}\right)$, $i=2,\ldots,n$. So
we have obtained a solution for the factor-coarsened system from a
solution for the original system. 

Dropping a point, say, $x_{2}^{1}$ is even simpler: we delete all
rows with $j_{1}=2$, and then redefine the $Q$ vector as
\[
\begin{array}{l}
Q'\left(l_{11},l_{13},\ldots,l_{1k_{1}},\ldots,l_{n1},\ldots,l_{nk_{n}}\right)\\
=\sum_{l_{12}}Q\left(l_{11},l_{12},l_{13},\ldots,l_{1k_{1}},\ldots,l_{n1},\ldots,l_{nk_{n}}\right).
\end{array}
\]

If the random variables involved have more than finite number of values
and/or the factors consist of more than finite number of factor points,
or if these numbers, though finite, are too large to handle the ensuing
linear programming problem, then LFT can still be used after the values
of the random variables and/or factors have been appropriately grouped.
LFT then becomes only a necessary condition for selective influences
(with respect to the original system of factors and random variables),
and its results will generally be different for different (non-nested)
groupings. 
\begin{example}
Consider the hypothesis $(A,B)\looparrowleft(\alpha,\beta)$ with
the factors having a finite number of factor points each, and $A$
and $B$ being response times. To use LFT, one can transform the random
variable $A$ as, say,
\[
A^{*}\left(\phi\right)=\begin{cases}
1 & \textnormal{if }A\left(\phi\right)\leq a_{1/4}\left(\phi\right),\\
2 & \textnormal{if }a_{1/4}\left(\phi\right)<A\left(\phi\right)\leq a_{1/2}\left(\phi\right),\\
3 & \textnormal{if }a_{1/2}\left(\phi\right)<A\left(\phi\right)\leq a_{3/4}\left(\phi\right),\\
4 & \textnormal{if }A\left(\phi\right)>a_{3/4}\left(\phi\right),
\end{cases}
\]
and transform $B$ as
\[
B^{*}\left(\phi\right)=\begin{cases}
1 & \textnormal{if }B\left(\phi\right)\leq b_{1/2}\left(\phi\right),\\
2 & \textnormal{if }B\left(\phi\right)>b_{1/2}\left(\phi\right),
\end{cases}
\]
where $a_{p}\left(\phi\right)$ and $b_{p}\left(\phi\right)$ designate
the $p$th quantiles of, respectively $A\left(\phi\right)$ and $B\left(\phi\right)$.
The initial hypothesis now is reformulated as $(A^{*},B^{*})\looparrowleft(\alpha,\beta)$,
with the understanding that if it is rejected then the initial hypothesis
will be rejected too (a necessary condition only). LFT will now be
applied to distributions of the form

\

\begin{center}%
\begin{tabular}{cc|cc|c|}
\hline 
\multicolumn{1}{|c}{$\alpha$} & $\beta$ & $A$ & $B$ & $\Pr$\tabularnewline
\hline 
\hline 
\multicolumn{1}{|c}{$x$} & y & 1 & 1 & $p_{11}$\tabularnewline
\cline{1-2} 
 &  & 1 & 2 & $p_{12}$\tabularnewline
 &  & $\vdots$ & $\vdots$ & $\vdots$\tabularnewline
 &  & 4 & 1 & $p_{41}$\tabularnewline
 &  & 4 & 2 & $p_{42}$\tabularnewline
\cline{3-5} 
\end{tabular}\end{center}

\

\protect{\noindent}where the marginals for $A$ are constrained to
0.25 and the marginals for $B$ to 0.5, for all treatments $\left\{ x^{\alpha},y^{\beta}\right\} $,
yielding a trivial compliance with marginal selectivity. Note that
the test may very well uphold $(A^{*},B^{*})\looparrowleft(\alpha,\beta)$
even if marginal selectivity is violated for $(A,B)(\phi)$ (e.g.,
if the quantiles $a_{p}\left(x^{\alpha},y^{\beta}\right)$ change
as a function of $y^{\beta}$). \qed
\end{example}

\subsection{\label{sub:Quantum-entanglement}Quantum entanglement}

Fine's (1982a-b) inequalities relate to the simplest EPR paradigm,
with the number of particles $n=2$, number of spin axes per particle
$k_{1}=k_{2}=2$, and the number of possible spin values per particle
$m_{1}=m_{2}=2$ (this value being the same for all spin axes chosen
for a given particle). They can be written, with reference to (\ref{eq:p's})
and (\ref{eq:p*'s}), as
\begin{align*}
-1\leq & P\left(2,2\,;\, j_{1},j_{2}\right)+P\left(2,2\,;\, j'_{1},j_{2}\right)\\
 & +P\left(2,2\,;\, j'_{1},j'_{2}\right)-P\left(2,2\,;\, j{}_{1},j'_{2}\right)\\
 & -P_{1}^{*}\left(\,2\,;\, j'_{1}\right)-P_{2}^{*}\left(\,2\,;\, j{}_{2}\right)\leq0,
\end{align*}
where $j_{1},j'_{1}\in\left\{ 1,2\right\} $, $j_{2},j'_{2}\in\left\{ 1,2\right\} $,
$j_{1}\not=j'_{1}$, $j_{2}\not=j'_{2}$. These inequalities constitute
the necessary and sufficient conditions for $\left(A_{1},A_{2}\right)\looparrowleft\left(\alpha_{1},\alpha_{2}\right)$,
with marginal selectivity assumed implicitly. Although Fine's derivation
of these inequalities is different, they can be derived as solutions
of system (\ref{eq:m*q=00003Dp*}), with $\mathbf{P}^{*}$ the 9-component
vector (using $\top$ for transposition) 
\[
\left(1,P_{1}^{*}\left(\,2\,;\,1\right),\ldots,P_{2}^{*}\left(\,2\,;\,2\right),P\left(2,2\,;\,1,1\right),\ldots,P\left(2,2\,;\,2,2\right)\right)^{\top},
\]
$\mathbf{Q}$ the 16-component vector
\[
\left(Q\left(1,1,1,1\right),\ldots,Q\left(2,2,2,2\right)\right)^{\top},
\]
and $\mathbf{M}^{*}$ the corresponding $9\times16$ Boolean matrix,

\[%
\begin{tabular}{|cc|cc|c|cccccccccccccccc|}
\cline{4-21} 
\multicolumn{1}{c}{} & \multicolumn{1}{c}{} & \multicolumn{1}{c|}{} & \multicolumn{2}{c|}{$H_{1^{\alpha}}$} & 1 & 1 & 1 & 1 & 1 & 1 & 1 & 1 & 2 & 2 & 2 & 2 & 2 & 2 & 2 & 2\tabularnewline
\multicolumn{1}{c}{} & \multicolumn{1}{c}{} & \multicolumn{1}{c|}{} & \multicolumn{2}{c|}{$H_{2^{\alpha}}$} & 1 & 1 & 1 & 1 & 2 & 2 & 2 & 2 & 1 & 1 & 1 & 1 & 2 & 2 & 2 & 2\tabularnewline
\multicolumn{1}{c}{} & \multicolumn{1}{c}{} & \multicolumn{1}{c|}{} & \multicolumn{2}{c|}{$H_{1^{\beta}}$} & 1 & 1 & 2 & 2 & 1 & 1 & 2 & 2 & 1 & 1 & 2 & 2 & 1 & 1 & 2 & 2\tabularnewline
\multicolumn{1}{c}{} & \multicolumn{1}{c}{} & \multicolumn{1}{c|}{} & \multicolumn{2}{c|}{$H_{2^{\beta}}$} & 1 & 2 & 1 & 2 & 1 & 2 & 1 & 2 & 1 & 2 & 1 & 2 & 1 & 2 & 1 & 2\tabularnewline
\hline 
$\alpha$ & $\beta$ & $A$ & $B$ & \multicolumn{1}{c}{} &  &  &  &  &  &  &  &  &  &  &  &  &  &  &  & \multicolumn{1}{c}{}\tabularnewline
\cline{1-4} \cline{6-21} 
$\cdot$ & $\cdot$ & $\cdot$ & $\cdot$ &  & 1 & 1 & 1 & 1 & 1 & 1 & 1 & 1 & 1 & 1 & 1 & 1 & 1 & 1 & 1 & 1\tabularnewline
1 & $\cdot$ & 2 & $\cdot$ &  & $0$ & $0$ & $0$ & $0$ & $0$ & $0$ & $0$ & $0$ & 1 & 1 & 1 & 1 & 1 & 1 & 1 & 1\tabularnewline
2 & $\cdot$ & 2 & \multicolumn{1}{c|}{$\cdot$} &  & $0$ & $0$ & $0$ & $0$ & 1 & 1 & 1 & 1 & $0$ & $0$ & $0$ & $0$ & 1 & 1 & 1 & 1\tabularnewline
$\cdot$ & 1 & $\cdot$ & \multicolumn{1}{c|}{2} &  & $0$ & $0$ & 1 & 1 & $0$ & $0$ & 1 & 1 & $0$ & $0$ & 1 & 1 & $0$ & $0$ & 1 & 1\tabularnewline
$\cdot$ & 2 & $\cdot$ & \multicolumn{1}{c|}{2} &  & $0$ & 1 & $0$ & 1 & $0$ & 1 & $0$ & 1 & $0$ & 1 & $0$ & 1 & $0$ & 1 & $0$ & 1\tabularnewline
1 & 1 & 2 & \multicolumn{1}{c|}{2} &  & $0$ & $0$ & $0$ & $0$ & $0$ & $0$ & $0$ & $0$ & $0$ & $0$ & 1 & 1 & $0$ & $0$ & 1 & 1\tabularnewline
1 & 2 & 2 & \multicolumn{1}{c|}{2} &  & $0$ & $0$ & $0$ & $0$ & $0$ & $0$ & $0$ & $0$ & $0$ & 1 & $0$ & 1 & $0$ & 1 & $0$ & 1\tabularnewline
2 & 1 & 2 & \multicolumn{1}{c|}{2} &  & $0$ & $0$ & $0$ & $0$ & $0$ & $0$ & 1 & 1 & $0$ & $0$ & $0$ & $0$ & $0$ & $0$ & 1 & 1\tabularnewline
2 & 2 & 2 & 2 &  & $0$ & $0$ & $0$ & $0$ & $0$ & 1 & $0$ & 1 & $0$ & $0$ & $0$ & $0$ & $0$ & 1 & $0$ & 1\tabularnewline
\cline{1-4} \cline{6-21} 
\end{tabular}\]

In fact, using a standard facet enumeration program (e.g., lrs program
at http://cgm.cs.mcgill.ca/\textasciitilde{}avis/C/lrs.html) these
inequalities (together with the equalities representing marginal selectivity)
can be derived ``mechanically.'' The essence of the computation
is in the fact that a linear system (\ref{eq:mq=00003Dp}) or (\ref{eq:m*q=00003Dp*})
is feasible if and only if the point $\mathbf{P}$ (respectively,
$\mathbf{P^{*}}$) belongs to the convex hull of the points corresponding
to the columns of $\mathbf{M}$ (respectively, $\mathbf{M}^{*}$),
which form a subset of the vertices of a unit hypercube. The facet
enumeration programs derive inequalities describing this convex hull. 

Given a set of numerical (experimentally estimated or theoretical)
probabilities, computing the LP problem (\ref{eq:mq=00003Dp}) or
(\ref{eq:m*q=00003Dp*}) is always preferable to dealing with explicit
inequalities as their number becomes very large even for moderate-size
vectors $P$. While Fine's inequalities for $n=2$, $k_{1}=k_{2}=2$,
$m_{1}=m_{2}=2$ (assuming marginal selectivity) number just 8, already
for $n=2$, $k_{1}=k_{2}=2$ with $m_{1}=m_{2}=3$ (describing, e.g.,
an EPR experiment with two spin-$1$ particles, or two spin-$\nicefrac{1}{2}$
ones and inefficient detectors), our computations yield 1080 inequalitiies,
and for $n=3$, $k_{1}=k_{2}=k_{3}=2$ and $m_{1}=m_{2}=m_{3}=2$,
corresponding to the Greenberger, Horne, \& Zeilinger (1989) paradigm
with three spin-$\nicefrac{1}{2}$ particles, this number is 53792.

The potential of JDC to lead to LFT and provide an ultimate criterion
for the entanglement problem has not been utilized in quantum physics
until relatively recently, when LFT was proposed in Werner \& Wolf
(2001a, b) and Basoalto \& Percival (2003). Prior to this, criteria
(as opposed to just necessary conditions) for the possibility of a
classical explanation for an EPR paradigm involving multiple particles,
multiple measurement settings, and multiple outcomes per measurements
were only known under strong symmetry constraints (de Barros \& Suppes,
2001; Garg, 1983; Mermin, 1990; Peres, 1999).

\subsection{Sample-level tests\label{sub:Sample-level-tests}}

\label{sub:stats}Although this paper is not concerned with statistical
questions, it may be useful to mention some of the approaches to constructing
sample-level tests based on LFT. As mentioned in Section \ref{sub:Quantum-entanglement},
the set of our vectors $\mathbf{P}$ for which the system $\mathbf{MQ}=\mathbf{P},\;\mathbf{Q}\geq0$
has a solution forms a convex polytope. In particular, if the set
$T$ of allowable treatments contains all combinations of factors
points, the polytope is the $\left(\left(k_{1}\left(m_{1}-1\right)+1\right)\ldots\left(k_{n}\left(m_{n}-1\right)+1\right)-1\right)\textnormal{-}$dimensional
convex hull of the points corresponding to the columns of the Boolean
matrix $\mathbf{M}$, which form a subset of the vertices of the $\left(m_{1}\right)^{k_{1}}\ldots\left(m_{n}\right)^{k_{n}}\textnormal{-}$dimensional
unit hypercube. Recently Davis-Stober (2009) developed a statistical
theory for testing the hypothesis that a vector of probabilities $\mathbf{P}$
(not necessarily of the same structure as in LFT) belongs to a convex
polytope $\mathcal{P}$ against the hypothesis that it does not. Under
certain regularity constraints he derived the asymptotic distribution
(a convex mixture of chi-square distributions) for the log maximum
likelihood ratio statistic
\[
-2\log\frac{\max_{\mathbf{P}\in\mathcal{P}}L\left(\mathbf{P}|\mathbf{N}\right)}{\max_{\mathbf{P}}L\left(\mathbf{P}|\mathbf{N}\right)},
\]
where $\mathbf{N}$ is the vector of observed absolute frequencies,
comprised of the numbers of occurrences of $\left(l_{1},\ldots,l_{n};j_{1},\ldots,j_{n}\right)$
in the case of LFT. The likelihoods $L\left(\mathbf{P}|\mathbf{N}\right)$
are computed using the standard theory of multinomial distributions.
This theory has been ``test-driven'' on the polytopes related to
the transitivity of preferences problem (Regenwetter, Dana, \& Davis-Stober,
2010, 2011). A Bayesian approach to the same problem is presented
in Myung, Karabatsos, \& Iverson (2005).

Other approaches readily suggest themselves. One of them is to use
the known theory of $L\left(\mathbf{P}|\mathbf{N}\right)/\max_{\mathbf{P}}L\left(\mathbf{P}|\mathbf{N}\right)$
to compute a confidence region of possible probability vectors $\mathbf{P}$
for a given empirical vector $\mathbf{N}$. The hypothesis of selective
influences is retained or rejected according as this confidence region
contains or does not contain a point $\mathbf{P}$ that passes LFT.
Resampling techniques is another obvious approach, e.g., the permutation
test in which the assignment of empirical distributions to different
treatments is randomly ``reshuffled'' so that each distribution
generally ends up assigned to a ``wrong'' treatment. If the proportion
of the permuted assignments whose deviation from the LFT polytope
does not exceed that of the the observed estimate of $\mathbf{P}$
is sufficiently small, the hypothesis of selective influences can
be considered supported. 

Little is known at present about the computational feasibility and
statistical properties of these and similar procedures. In particular
(this also applies to Davis-Stober's test), we do not know their statistical
power for different locations of the true vector of probabilities
outside the convex polytope described by $\mathbf{MQ}=\mathbf{P},\;\mathbf{Q}\geq0$.
Nor do we know how the effect size, a measure of deviation of $\mathbf{P}$
from the polytope, should be computed optimally. All of this will
have to be investigated separately.

\section{Conclusion}

Selectiveness in the influences exerted by a set of inputs upon a
set of random and stochastically interdependent outputs is a critical
feature of many psychological models, often built into the very language
of these models. We speak of an internal representation of a given
stimulus, as separate from an internal representation of another stimulus,
even if these representations are considered random entities and they
are not independent. We speak of decompositions of response time into
signal-dependent and signal-independent components, or into a perceptual
stage (influenced by stimuli) and a memory-search stage (influenced
by the number of memorized items), without necessarily assuming that
the two components or stages are stochastically independent.

In this paper, we have described the Linear Feasibility Test, an application
of the fundamental Joint Distribution Criterion for selective influences
to random variables with finite numbers of values. This test can be
performed by means of standard linear programming. Due to the fact
that any random output can be discretized, the Linear Feasibility
Test is universally applicable, although one should keep in mind that
if a diagram of selective influences is upheld by the test at some
discretization, it may be rejected at a finer or non-nested discretization
(but not at a coarser one). Both the Joint Distribution Criterion
and the Linear Feasibility Test, although new in the behavioral context,
have their direct analogues in quantum physics, in dealing with the
problem of the existence of a classical explanation (one with non-contextual,
local hidden variables) for outcomes of noncommuting measurements
performed on entangled particles. The discovery of these parallels
promises to enrich and facilitate our understanding of selective influences.

\section*{Acknowledgments}

This research has been supported by AFOSR grant FA9550-09-1-0252 to
Purdue University and by the Academy of Finland grant 121855 to University
of Jyväskylä. We are indebted to Joseph Houpt and Jerome Busemeyer
whose comments helped us to significantly improve the paper.

\section*{REFERENCES}

\setlength{\parindent}{0cm}\everypar={\hangindent=15pt}Basoalto,
R.M., \& Percival, I.C. (2003). BellTest and CHSH experiments with
more than two settings. \emph{Journal of Physics A: Mathematical \&
General}, 36, 7411\textendash{}7423.

Bell, J. (1964). On the Einstein-Podolsky-Rosen paradox. \emph{Physics},
1, 195-200.

Bohm, D., and Aharonov, Y. (1957). Discussion of Experimental Proof
for the Paradox of Einstein, Rosen and Podolski. \emph{Physical Review},
108, 1070-1076.

Clauser, J. F., Horne, M. A., Shimony, A. \& Holt, R. A. (1969). Proposed
experiment to test local hidden-variable theories.\emph{ Physical
Review Letters}, 23, 880\textendash{}884.

Davis-Stober, C. P. (2009). Analysis of multinomial models under inequality
constraints: Applications to measurement theory. \emph{Journal of
Mathematical Psychology}, 53, 1\textendash{}13.

de Barros, J. A., \& Suppes, P. (2001). Results for six detectors
in a three-particle GHZ experiment\emph{. }In J. Bricmont, D. Durr,
M. C. Galavotti, G. Ghirardi, F. Petruccione, N. Zanghi (Eds.)\emph{
Chance in Physics: Foundations and Perspectives} (pp. 213-223), Berlin:
Springer.

Dzhafarov, E.N. (1999). Conditionally selective dependence of random
variables on external factors. \emph{Journal of Mathematical Psychology},
43, 123-157.

Dzhafarov, E.N. (2003a). Selective influence through conditional independence.
\emph{Psychometrika}, 68, 7\textendash{}26.

Dzhafarov, E.N., \& Gluhovsky, I. (2006). Notes on selective influence,
probabilistic causality, and probabilistic dimensionality. \emph{Journal
of Mathematical Psychology}, 50, 390\textendash{}401.

Dzhafarov, E.N., \& Kujala, J.V. (2010). The Joint Distribution Criterion
and the Distance Tests for selective probabilistic causality. \emph{Frontiers
in Quantitative Psychology and Measurement}, 1:151 doi: 10.3389/fpsyg.2010.00151.

Einstein, A, Podolsky, B., \& Rosen, N. (1935). Can quantum-mechanical
description of physical reality be considered complete? \emph{Physical
Review}, 47, 777\textendash{}780.

Fine, A. (1982a). Joint distributions, quantum correlations, and commuting
observables. \emph{Journal of Mathematical Physics}, 23, 1306-1310.

Fine, A. (1982b). Hidden variables, joint probability, and the Bell
inequalities.\emph{ Physical Review Letters}, 48, 291-295.

Garg, A. (1983). Detector error and Einstein-Podolsky-Rosen correlations.
\emph{Physical Review D}, 28, 785-790.

Greenberger, D.M., Horne, M.A., \& Zeilinger, A. (1989). Going beyond
Bell's theorem. In M. Kafatos (Ed.\emph{) Bell\textquoteright{}s Theorem,
Quantum Theory and Conceptions of the Universe} (pp. 69\textendash{}72),
Dordrecht: Kluwer.

Karmarkar, N. (1984). A new polynomial-time algorithm for linear programming.\emph{
}Combinatorica, 4, 373-395.

Khachiyan, L. (1979). A polynomial algorithm in linear programming.
Doklady Akademii Nauk SSSR 244, 1093-1097.

Kechris, A. S. (1995). \emph{Classical Descriptive Set Theory}. New
York: Springer.

Kujala, J. V., \& Dzhafarov, E. N. (2008). Testing for selectivity
in the dependence of random variables on external factors. \emph{Journal
of Mathematical Psychology}, 52, 128\textendash{}144.

Mermin, N.D. (1990). Extreme quantum entanglement in a superposition
of macroscopically distinct states\emph{. Physical Review Letters},
65, 1838-1840.

Myung, J. I., Karabatsos, G., \& Iverson, G. J. (2005). A Bayesian
approach to testing decision making axioms. \emph{Journal of Mathematical
Psychology}, 49, 205\textendash{}225.

Peres, A. (1999). All the Bell inequalities. \emph{Foundations of
Physics}, 29, 589-614.

Regenwetter, M., Dana, J., \& Davis-Stober, C. P. (2010). Testing
transitivity of preferences on two-alternative forced choice data.
\emph{Frontiers in Quantitative Psychology and Measurement}, doi:
10.3389/fpsyg.2010.00148.

Regenwetter, M., Dana, J., Davis-Stober, C. P. (2011). Transitivity
of preferences. \emph{Psychological Review}, 118, 42-56.

Sternberg, S. (1969). The discovery of processing stages: Extensions
of Donders\textquoteright{} method. In W.G. Koster (Ed.), \emph{Attention
and Performance II. Acta Psychologica}, 30, 276\textendash{}315.

Stapp, H.P. (1975). Bell\textquoteright{}s Theorem and World Process\emph{.
}Nuovo Cimento B 29, 270-276.

Suppes, P., \& Zanotti, M. (1981). When are probabilistic explanations
possible? \emph{Synthese}, 48, 191\textendash{}199.

Thurstone, L. L. (1927). A law of comparative judgments. \emph{Psychological
Review}, 34, 273\textendash{}286.

Townsend, J. T. (1984). Uncovering mental processes with factorial
experiments. \emph{Journal of Mathematical Psychology}, 28, 363\textendash{}400. 

Townsend, J.T., \& Schweickert, R. (1989). Toward the trichotomy method
of reaction times: Laying the foundation of stochastic mental networks.
\emph{Journal of Mathematical Psychology}, 33, 309\textendash{}327.

Werner, R.F., \& Wolf, M.M. (2001). All multipartite Bell correlation
inequalities for two dichotomic observables per site. arXiv:quant-ph\slash{}0102024v1.

Werner, R.F., \& Wolf, M.M. (2001). Bell inequalities and entanglement.
arXiv:quant-ph\slash{}0107093v2.
\end{document}